\documentclass[twocolumn,american,aps, pra, floatfix, table, dvipsnames, superscriptaddress]{revtex4-2}
\usepackage[T1]{fontenc}
\usepackage[latin9]{inputenc}
\usepackage{color}
\usepackage{booktabs}
\usepackage{mathrsfs}
\usepackage{amsmath}
\usepackage{amssymb}
\usepackage{amsthm}
\usepackage{float}  
\usepackage{graphicx}
\usepackage{esint}
\usepackage[bookmarks=false,
 breaklinks=true,pdfborder={0 0 1},backref=false,colorlinks=true]
 {hyperref}
\hypersetup{
 pdfcreator={},pdfproducer={LaTeX with hyperref},linkcolor=blue,anchorcolor=blue,citecolor=blue,filecolor=red,menucolor=red,pagecolor=red,urlcolor=blue,pdfstartview=FitV,pdfhighlight=/I, hypertexnames=true}

\usepackage{newtxtext, xspace}
\usepackage{bbm}
\usepackage[table, dvipsnames]{xcolor}
\usepackage{pgfplots}
\usepackage{mathtools}
\usepackage{algpseudocode}
\usepackage{bm} 
\usepackage[ruled,vlined]{algorithm2e}

\usepackage{multirow} 

\usepackage{natbib}
\setcitestyle{numbers}

\usepackage{longtable}

\usepackage{siunitx} 

\newtheorem{theorem}{Theorem}
\newtheorem{lemma}[theorem]{Lemma}
\newtheorem{proposition}[theorem]{Proposition}

\begin{document}
\title{
Attributed-graphs kernel implementation using local detuning of neutral-atoms Rydberg Hamiltonian
}

\author{Mehdi Djellabi}
\author{Matthias Hecker}
\author{Shaheen Acheche}
\affiliation{Pasqal, 24 avenue Emile Baudot, 91120 Palaiseau}

\date{\today }
\begin{abstract}
We extend the quantum-feature kernel framework, which relies on measurements of graph-dependent observables, along three directions. 
First, leveraging neutral-atom quantum processing units (QPUs), we introduce a scheme that incorporates attributed graphs by embedding edge features into atomic positions and node features into local detuning fields of a Rydberg Hamiltonian. 
We demonstrate both theoretically and empirically that local detuning enhances kernel expressiveness. 
Second, in addition to the existing quantum evolution kernel ($\mathrm{QEK}$), which uses global observables, we propose the generalized-distance quantum-correlation ($\mathrm{GDQC}$) kernel, based on local observables. 
While the two kernels show comparable performance, we show that $\mathrm{GDQC}$ can achieve higher expressiveness. 
Third, instead of restricting to observables at single time steps, we combine information from multiple stages of the quantum evolution via pooling operations. 
Using extensive simulations on two molecular benchmark datasets, MUTAG and PTC\_FM, we find: (a) $\mathrm{QEK}$ and $\mathrm{GDQC}$ perform competitively with leading classical algorithms; and (b) pooling further improves performance, enabling quantum-feature kernels to surpass classical baselines. 
These results show that node-feature embedding and kernel designs based on local observables advance quantum-enhanced graph machine learning on neutral-atom devices.

\end{abstract}

\maketitle
\section{Introduction\label{sec:introduction}}

Graph-structured data arise in domains as diverse as social media~\cite{Sharma_social}, chemistry~\cite{Gilmer:2017arh}, and program analysis~\cite{allamanis2018learning}, yet their intrinsic metric (shortest path distance) is discrete and non-Euclidean, rendering classical vector-space similarities inadequate. 
Furthermore, exact combinatorial measures such as graph isomorphism or graph-edit distance offer perfect discrimination, but even with state-of-the-art heuristic accelerations they become computationally prohibitive as graphs grow larger~\cite{Xu_edit}. 
To obtain tractable and differentiable notions of similarity, the field has turned to graph kernels, which embed graphs into a reproducing-kernel Hilbert space where powerful kernel methods apply. 
The seminal survey of~\cite{vishwanathan10a} unified random-walk and other early kernels, while the Weisfeiler-Leman family leveraged color refinement to scale sub-tree comparisons to large graphs~\cite{shervashidze11a}. 
Complementary work casts graph comparison as an optimal-transport problem aligning intra-graph metric structures and gracefully handles graphs of differing sizes, proving effective for biological and social networks~\cite{Peyre_Gromov}. 
More precisely, these developments fit within the broader program of geometric deep learning, which generalizes representation-learning techniques to non-Euclidean domains~\cite{Bronstein_geometric}. 
\newline More recently, interest has surged in mapping graphs to quantum systems and computing similarities directly from physical measurements. Early work framed the evolution of a continuous-time quantum walk as a positive-definite kernel on graphs~\cite{continuous_RW,Thabet_Quantum}, and subsequent extensions employed discrete-time walks and information-theoretic divergences to enrich the feature space~\cite{BAI2015344}.
Modern quantum platforms now encode each graph as a parameterised quantum state and estimate kernel entries via repeated state-overlap measurements~\cite{Morgado_2021}, a strategy validated experimentally on a 10-qubit nuclear-magnetic-resonance processor~\cite{Sabarad:2024tgy}.
Complementing these hardware demonstrations, industry toolkits such as Pasqal's open-source Quantum Evolution Kernel ($\mathrm{QEK}$)~\cite{Henry_2021, Albrecht2023} compute graph kernels from expectation values~\cite{Mukherjee_2025} of graph-dependent observables measured on neutral-atom devices~\cite{qek_repo}. 
Together, these developments suggest that quantum hardware may soon deliver graph kernels with higher expressive power and lower sample complexity than their classical counterparts, reinforcing the motivation for the methods introduced here. 
Despite this notable progress, the graph-kernel literature from the perspective of quantum computing remains strongly skewed towards unattributed graphs. Kernels that focus on attributed graphs with discrete labels attached to vertices and edges have been comparatively few~\cite{Bai_attributed}. 
 
Quantum feature kernels such as used in the $\mathrm{QEK}$ algorithm represent a particular class of quantum kernels within which the quantum features are derived from measured observables. 
Herein, a graph $\mathcal{G}$ is embedded into the system Hamiltonian $\hat{H}^{\mathcal{G}}$ such that a time evolution imposes the graph information onto the 
quantum state $\left|\psi_{\mathcal{G}}\left(t\right)\right\rangle $, and eventually onto any measured observable $O_{\alpha}^{\mathcal{G}}\left(t\right)=\left\langle \psi_{\mathcal{G}}\left(t\right)\right|\hat{O}_{\alpha}\left|\psi_{\mathcal{G}}\left(t\right)\right\rangle $.
For the purpose of a quantum feature kernel, multiple observables can be considered ($\alpha=1,\dots,N_{M}$) and measurements can be taken at multiple time steps ($t_1,\dots,t_{N_t}$).
In combination, they provide the $\left(N_{M} N_{t}\right)$-dimensional
quantum-feature vector $\boldsymbol{O}^{\mathcal{G}}$ which encodes the graph information in a non-trivial way.
This allows to generically define a quantum feature kernel between two graphs of a given dataset, $\mathcal{G}_{\mu}$ and $\mathcal{G}_{\mu^\prime}$, as
\begin{align}
K_{\mu\mu^{\prime}} & =\kappa\left(\mathcal{G}_{\mu},\mathcal{G}_{\mu^{\prime}}\right)=F_{\kappa}\left[\boldsymbol{O}^{\mathcal{G}_{\mu}},\boldsymbol{O}^{\mathcal{G}_{\mu^{\prime}}}\right].\label{eq:general_quantum_evolution_kernel}
\end{align}
The kernel function $F_{\kappa}$ can be chosen in many ways. It only needs to ensure kernel
properties like positive semi-definiteness (PSD), \textit{i.e.} all eigenvalues of
$K$ have to be non-negative. 

In this work, we expand on the set of attributed-graphs tools by introducing an implementation scheme that allows quantum feature kernels (\ref{eq:general_quantum_evolution_kernel}) to deal with node attributes.
Specifically, we used the local detuning of a neutral-atom Rydberg Hamiltonian to encode node features.
We demonstrate that the corresponding feature kernels (\ref{eq:general_quantum_evolution_kernel}) are more expressive than without encoded node features, and we present an extensive numerical comparison between both scenarios on two well-known benchmark datasets, MUTAG and PTC\_FM. 
Furthermore, this work includes a study of different kernel functions $F_{\kappa}$, in which two conceptually different implementations are tested, namely $\mathrm{QEK}$ and the Generalized-Distance Quantum-Correlation ($\mathrm{GDQC}$) kernel.
In contrast to the established $\mathrm{QEK}$ kernel which exploits a global observable, the proposed $\mathrm{GDQC}$ kernel draws from a local observable, the correlation matrix, and is inspired by the classical generalized-distance Weisfeiler-and-Leman (GD-WL) kernel~\cite{shervashidze11a}.
Lastly, we add one more layer of enhanced expressive power to the project: pooling. 
Specifically, we combine kernels computed by observables measured at multiple time steps by means of pooling operations. 
Hereby, two different pooling schemes are tested, sum pooling and product pooling.
Having trained a support vector machine (SVM) using all these different kernel implementations, the results of the extensive numerical tests are presented in section~\ref{sec:expe}, and discussed.
We observe a trend that the kernels using node features generally yield better SVM predictions compared to the node-agnostic ones.
Among the quantum-feature kernels, $\mathrm{QEK}$ and $\mathrm{GDQC}$, we observe competitive performances with the classical baselines.
However, we find that pooling can enhance their performance and enable them to outperform the classical algorithms.

This paper is organized as follows:
In Sec.~\ref{sec:Quantum-graph-encoding} we detail all aspects of the graph-feature extraction. 
We introduce the neutral-atom Rydberg Hamiltonian (\ref{subsec:Rydberg_Hamiltonian}), introduce the attributed-graphs embedding scheme (\ref{subsec:Graph_Embedding}), present all details concerning the emulation (\ref{subsec:emulation}), and show analytical small-time approximations useful for performance interpretation (\ref{subsec:perturbation_theory}).
In Sec.~\ref{sec:kernels} we discuss the implemented kernel methods used in this work.
First, we revisit classical kernel methods (\ref{subsec:classical_kernels}), before we introduce the $\mathrm{QEK}$ kernel (\ref{subsec:kernel_QEK}), and the proposed $\mathrm{GDQC}$ kernel (\ref{subsec:algo}), and lastly in Subsec.~\ref{subsec:pooling}, we define the two implemented pooling schemes.
In Sec.~\ref{sec:expe}, the numerical results are presented and discussed. 
Finally, we provide a conclusion in Sec.~\ref{sec:conclusion}.
In App.~\ref{app:perturbation theory} we derive the analytical small-time approximations via time-dependent perturbation theory.
App.~\ref{app:classical_graph_kernel} provides practical details on the classical-kernel implementations.
In App.~\ref{app:proof_prop1} we rigorously demonstrate that the kernels with node-feature embedding are more expressive than the ones without. 
Apps.~\ref{app:GDQC} and~\ref{app:expressivity_gdqc} prove, respectively, that the $\mathrm{GDQC}$ kernel satisfies all kernel properties and that it is at least as expressive as the GD-WL kernel. 
In App.~\ref{app:kernel_pooling} we prove that the inclusion of different time steps can only increase the kernel expressiveness.

\section{Quantum feature extraction}\label{sec:Quantum-graph-encoding}

In this section we demonstrate how the quantum features $\boldsymbol{O}^{\mathcal{G}}$ [\textit{cf.} Eq.~(\ref{eq:general_quantum_evolution_kernel})] are extracted using a neutral-atom QPU. 
After introducing the Rydberg Hamiltonian (\ref{subsec:Rydberg_Hamiltonian}), we discuss the graph embedding scheme (\ref{subsec:Graph_Embedding}), and explain the details of the emulation used to extract the graph-dependent observables (\ref{subsec:emulation}).
Finally, we present small-time analytical approximations useful for later results` interpretation (\ref{subsec:perturbation_theory}).

\subsection{Rydberg / Ising Hamiltonian\label{subsec:Rydberg_Hamiltonian}}

Neutral-atom quantum computing processors (QPU) are operated such
that the two states describing a qubit, the zero state $\left|0\right\rangle $
and the one-state $\left|1\right\rangle $, correspond to particularly
chosen long-lived atomic energy levels where specifically the state
$\left|1\right\rangle $ is associated with a highly-excited atomic Rydberg
level. For example, for $^{87}\mathrm{Rb}$ atoms, commonly used in
the field, one such choice is given through $\left|0\right\rangle =\left|5S_{1/2},F=2,m_{F}=2\right\rangle $
and $\left|1\right\rangle =\left|60S_{1/2},m_{J}=1/2\right\rangle $.
This choice leads to an effective Ising Hamiltonian where the long-ranged
interaction is driven through the van-der-Waals interaction of distant
atoms which are simultaneously in the highly-excited and spatially
extended Rydberg state. As a consequence, a system of $N$ atoms is
described through the so-called Rydberg Hamiltonian~\cite{Gallagher_1994}
\begin{equation}
\hat{\mathcal{H}}\left(\mathbb{J},\Omega,\boldsymbol{\delta}\right)=\frac{1}{2}\sum\nolimits_{i,j}^{\prime}J_{ij}\hat{n}_{i}\hat{n}_{j}+\hbar\sum_{i=1}^{N}\left(\frac{\Omega}{2}\hat{\sigma}_{i}^{x}-\delta_{i}\hat{n}_{i}\right),\label{eq:Rydberg_Ham}
\end{equation}
where the summation $\sum\nolimits_{i,j}^{\prime}=\sum_{i,j=1\,|j\neq i}^{N}$
is restricted to non-identical indices. The coupling constant $J_{ij}=C_{6}/r_{ij}^{6}$
describes the strength of the van-der-Waals interaction between atoms
$i$ and $j$ which are separated by a distance $r_{ij}$, and $\mathbb{J}$
denotes the respective $N\times N$-dimensional matrix with $\mathbb{J}_{ij}=J_{ij}$.
The positive constant $C_{6}$ depends on the Rydberg level, and in
the chosen configuration its value is $C_{6}/\hbar=865723\,\mathrm{rad}\,\mathrm{\mu s}^{-1}\mathrm{\mu m}^{6}$ (corresponding to the Rydberg level $n=60$).
The external fields, the Rabi frequency $\Omega$ and the spatially-dependent
detuning $\boldsymbol{\delta}=\left(\delta_{1},\dots,\delta_{N}\right)$,
are primarily determined by the laser frequencies and intensities which
drive the atomic system. 
Throughout this work, we set the Rabi frequency to $\Omega = 2\pi \, \mathrm{rad}\,/\mathrm{\mu s} $.
The Hilbert space which the Hamiltonian (\ref{eq:Rydberg_Ham})
acts on is $2^{N}$ dimensional. Using the occupational (or measurement)
basis, a basis vector is given through $\left|\boldsymbol{b}\right\rangle =\left|b_{1}b_{2}\cdots b_{N}\right\rangle $
where each entry takes on values $b_{i}\in\left\{ 0,1\right\} $.
The action of the Pauli-matrix operators in (\ref{eq:Rydberg_Ham})
on a basis state is given by $\hat{n}_{i}\left|b_{i}\right\rangle =b_{i}\left|b_{i}\right\rangle $,
$\hat{\sigma}_{i}^{x}\left|b_{i}\right\rangle =\left|1-b_{i}\right\rangle $
and $\hat{\sigma}_{i}^{z}\left|b_{i}\right\rangle =\left(-1\right)^{b_{i}}\left|b_{i}\right\rangle $
given that also $\hat{n}_{i}=\big(1-\hat{\sigma}_{i}^{z}\big)/2$.
The last relation can be used to translate the Hamiltonian (\ref{eq:Rydberg_Ham})
into the standard Ising representation 
\begin{equation}
\hat{\mathcal{H}}\left(\mathbb{J},\Omega,\boldsymbol{\delta}\right)=\frac{1}{8}\sum\nolimits_{i,j}^{\prime}J_{ij}\hat{\sigma}_{i}^{z}\hat{\sigma}_{j}^{z}+\hbar\sum_{i=1}^{N}\left(\frac{\Omega}{2}\hat{\sigma}_{i}^{x}+h_{i}^{z}\hat{\sigma}_{i}^{z}\right)+\bar{E},\label{eq:Ising_Ham}
\end{equation}
where the constant $\bar{E}=\mathrm{tr}\big\{\hat{\mathcal{H}}\big\} / 2^{N}=\frac{1}{8}\sum\nolimits_{i,j}^{\prime}J_{ij}-\frac{1}{2}\hbar\sum_{i}\delta_{i}$
denotes the average energy, and the longitudinal field is given by
\begin{align}
h_{i}^{z} & =\frac{1}{2}\delta_{i}-\frac{1}{4\hbar}\sum_{j\neq i}J_{ij}.\label{eq:h_z}
\end{align}
The Ising Hamiltonian (\ref{eq:Ising_Ham}) with its more symmetric
representation can be more revealing for a physical intuition,
as we will demonstrate in part~\ref{subsec:Graph_Embedding}.

By design, the initial state in a neutral-atom QPU is the all-zero
state $\left|\psi\left(0\right)\right\rangle =\left|\boldsymbol{0}\right\rangle $.
With the Hamiltonian (\ref{eq:Rydberg_Ham}) being time independent,
the time evolution corresponds to a quench experiment, whereby the
time-evolved state is given by 
\begin{align}
\left|\psi\left(t\right)\right\rangle  & =\exp\left[-\mathsf{i}(t/\hbar)\,\hat{\mathcal{H}}\left(\mathbb{J},\Omega,\boldsymbol{\delta}\right)\right]\,\left|\psi\left(0\right)\right\rangle .\label{eq:psi_t}
\end{align}
The interest of the current work lies in three specific observables which will serve as the quantum features $\boldsymbol{O}^{\mathcal{G}}$ [\textit{cf.} Eq. (\ref{eq:general_quantum_evolution_kernel})]: the excitation number
$n_{i}\left(t\right)$, the correlation matrix $C_{ij}\left(t\right)$,
and the probability for a fixed number $k$ of excitations $P_{k}\left(t\right)$. These observables
are defined as
\begin{align}
P_{k}\left(t\right) & =\left\langle \psi\left(t\right)\big|\hat{P}_{k}\big|\psi\left(t\right)\right\rangle , & n_{i}\left(t\right) & =\left\langle \psi\left(t\right)\big|\hat{n}_{i}\big|\psi\left(t\right)\right\rangle ,\nonumber \\
C_{ij}\left(t\right) & =\left\langle \psi\left(t\right)\big|\hat{n}_{i}\hat{n}_{j}\big|\psi\left(t\right)\right\rangle ,\label{eq:observables}
\end{align}
where the projection operator onto a fixed-excitation subspace is defined
as $\hat{P}_{k}=\sum_{\boldsymbol{b}}\delta_{k,\sum_{i=1}^Nb_{i}}\,\left|\boldsymbol{b}\right\rangle \left\langle \boldsymbol{b}\right|$
with the Kronecker delta $\delta_{k,k^{\prime}}$ and $k=0,1,\dots,N$.
Note that since $\hat{n}_{i}^2 = \hat{n}_{i}$, the diagonal elements of the correlation matrix $C_{ii}\left(t\right)=n_{i}\left(t\right)$ are identical to the excitation numbers $n_{i}\left(t\right)$, also referred to as the site occupations.
All observables (\ref{eq:observables}) are diagonal in the measurement
basis, and thus, they can be easily sampled after each run of the neutral-atom QPU. 

The Rydberg blockade mechanism~\cite{Browaeys2020, Gaetan_2009} provides a useful estimate of the length scale over which Rydberg-Rydberg interactions suppress simultaneous excitation of atoms. 
Equating the van-der-Waals interaction shift to the drive energy scale, $C_6/R_\mathrm{b}^6=\hbar \Omega$, defines the blockade radius $R_\mathrm{b}=(C_6/\hbar \Omega)^{1/6}$. 
For interatomic separations $r_{ij}\lesssim R_\mathrm{b}$, excitation of the doubly excited state is energetically suppressed (shifted out of resonance).
In our setting, we choose $R_\mathrm{b}=7.19\,\mathrm{\mu m}$ with a nearest-neighbor distance $r_\mathrm{NN}=5.01\,\mathrm{\mu m}$ (see below) such that $r_\mathrm{NN}<R_\mathrm{b}<2r_\mathrm{NN}$, which blockades nearest neighbors while leaving neighbors at further distances largely unblocked.

\subsection{Attributed-graphs embedding\label{subsec:Graph_Embedding}}

In the present work we deal with graphs that originate from two common
benchmark datasets for binary classification: mutagenic aromatic and
heteroaromatic nitro compounds (MUTAG) consisting of a total of $188$
molecules with the task of predicting whether a molecule is mutagenic on Salmonella typhimurium or not, and predictive toxicology challenge -- female mice
(PTC\_FM) consisting of $349$ molecules with the task of predicting
if a compound is carcinogenic~\cite{MUTAG_1991, PTCFM_2003, TUD_dataset}. In both datasets the compounds
are provided as attributed graphs $\mathcal{G}=\left(\mathcal{V},\mathcal{E},\mathcal{L}\right)$, also referred to as labeled graphs or heterogeneous graphs,
with a set of nodes $\mathcal{V}$ (the atoms), a set of edges
$\mathcal{E}$ (the bonds between atoms), and a set
of node labels $\mathcal{L}$ (the name of the
atoms) where the latter are interchangeably referred to as labels, attributes or features.

A common strategy to embed a graph onto the Hamiltonian (\ref{eq:Rydberg_Ham})
is by choosing the atom`s positions such that connected atoms have
a significant coupling strength $J_{ij}$ whereas unconnected atoms
are coupled with a negligible strength. Mathematically, we introduce
the graph-dependent coupling matrix as 
\begin{align}
\mathbb{J}_{ij}^{\mathcal{G}} & =\begin{cases}
J_{\mathrm{NN}} & ,\left(i,j\right)\in\mathcal{E}\\
J_{ij}\ll J_{\mathrm{NN}} & ,\left(i,j\right)\notin\mathcal{E}
\end{cases},\label{eq:graph_dependent_J}
\end{align}
where we defined the nearest-neighbor coupling strength $J_{\mathrm{NN}}=C_{6}/r_{\mathrm{NN}}^{6}$
with $r_{\mathrm{NN}}$ being the nearest-neighbor distance. 
Throughout this study, we fix the nearest-neighbor distance to $r_{\mathrm{NN}}=5.01\,\mathrm{\mu m}$
such that the nearest-neighbor coupling becomes $J_{\mathrm{NN}}/\hbar\approx 54.75\,\mathrm{rad} / \mu s$.
The choice in Eq.~(\ref{eq:graph_dependent_J}) is appealing for graph embedding, although it restricts this scheme to a specific class of graphs, namely unit-disk graphs. 
Here, an edge between two nodes exists only if their distance is below a fixed threshold distance~\cite{CLARK1990165}.
In general, the task of finding a unit-disk graph embedding is a NP-complete problem~\cite{Kuhn_2004}.
Conveniently, the molecules in the dataset just like the atoms in the QPU interact through physical forces which naturally makes them unit-ball graphs.
The only problem is the dimension.
Current neutral-atom QPU only allow for two-dimensional registers. 
Therefore, in our study we only include molecules that admit a two-dimensional embedding reducing the size of PTC\_FM dataset to $332$ molecules, corresponding to $95\%$ of the original dataset. 
The MUTAG dataset was kept intact, since we obtained a two-dimensional representation of all 188 constituent molecules.
We denote the reduced PTC\_FM dataset as PTC\_FM$^*$.
The algorithm used to compute the layouts is an extension of the one described in Ref.~\cite{Albrecht2023} (appendix E therein), \textit{i.e.} a force-based algorithm~\cite{Blade, Tamassia-chapter12}. 

\begin{figure*}[!htbp]

\includegraphics[width=\textwidth]{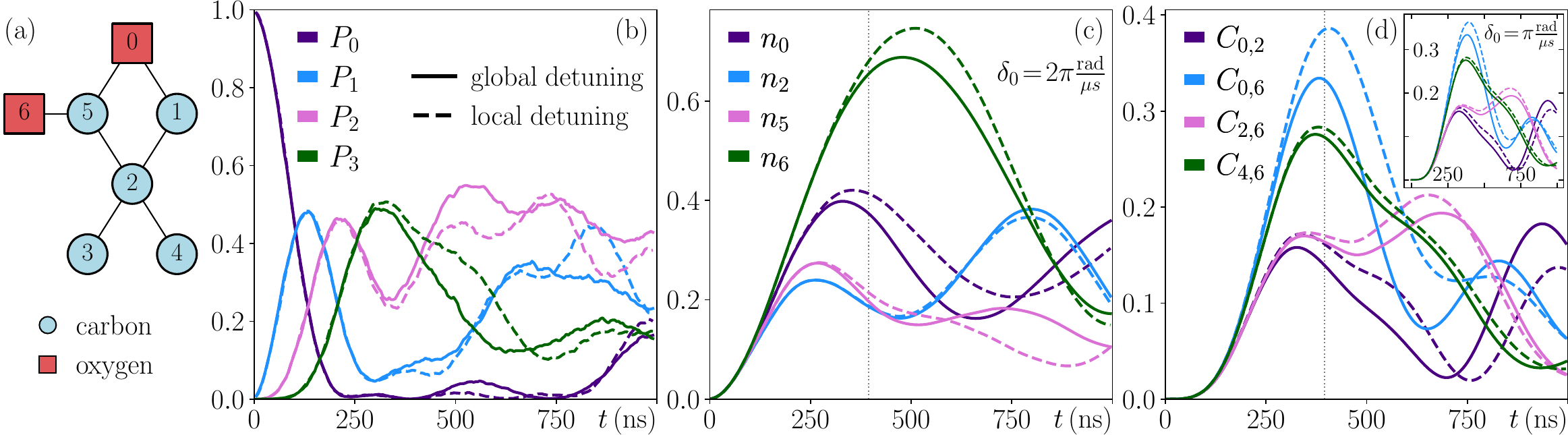}
\caption{Emulation results for the molecule in panel (a). 
Panels (b)-(d) display the time evolution of the three observables defined in Eq.~(\ref{eq:observables}): the excitation probability $P_k$, the site occupations $n_i = \langle \hat{n}_i \rangle$, and the pairwise correlations $C_{i,j} = \langle \hat{n}_i  \hat{n}_j \rangle$. 
Solid (dashed) curves correspond to evolution under the global-detuning (local-detuning) Hamiltonian, with detuning scale set to $\delta_{0}=2\pi\,\mathrm{rad}\,/\mathrm{\mu s}$; see Eqs.~(\ref{eq:H_global}) and (\ref{eq:H_local}).
For comparison, the inset in (d) shows results for a smaller detuning scale, $\delta_{0}=\pi\,\mathrm{rad}\,/\mathrm{\mu s}$. 
For the molecule and parameters shown, differences between global and local detuning become apparent by $t \gtrsim 250 \, \mathrm{ns}$. 
The molecule in (a) is taken from the PTC\_FM$^*$ dataset and has SMILES string "CC1(C)COC1=O". 
The gray dotted vertical line marks $t=0.4\,\mathrm{\mu s}$, which is the time step used for the visualization in Fig.~\ref{fig:GDQC_illustration}.
\label{fig:dynamics}}
\end{figure*}

The graph-dependent coupling matrix $\mathbb{J}_{ij}^{\mathcal{G}}$ (\ref{eq:graph_dependent_J}) only
embeds the edges $\mathcal{E}$ of the graph, but it makes no use of the node features $\mathcal{L}$.
The embedding of the node features $\mathcal{L}$ is a key element of this work.
It is achieved through an appropriate choice of the local detuning $\delta_{i}\rightarrow\delta_{i}^{\mathcal{G}}$, see Eq.~(\ref{eq:Rydberg_Ham}).
Specifically, we have chosen a scheme that exploits the atomic masses
$m_{i}$ within the molecules, and we define 
\begin{align}
\delta_{i}^{\mathcal{G}} & =\delta_{0} \left(1 - \frac{m_C}{m_i}\right)
\label{eq:local_detuning}\end{align}
where $\delta_{0}$ is a fixed offset value and $m_{C}=12\,\mathrm{u}$ is
the mass of a carbon atom, which is the lightest element present in our datasets (hydrogen atoms are omitted in the graphs.). 
The atomic masses are obtained from PubChem Periodic Table webpage~\cite{pubchem_periodic_table}. 
The local detuning choice (\ref{eq:local_detuning}) is motivated as follows. 
In earlier $\mathrm{QEK}$ work~\cite{Albrecht2023} the global detuning was set to zero. 
To ensure that the local detuning provides additional expressivity beyond this global baseline, we design Eq.~(\ref{eq:local_detuning}) such that the detuning remains zero for the most abundant element in the datasets. 
This is the case for carbon, which accounts for approximately $70\%$ of all atoms. 
Consequently, only the remaining atoms receive a nonzero local detuning, i.e., $\delta_i^{\mathcal{G}} > 0$ whenever $m_i > m_C$.
The scheme in Eq.~(\ref{eq:local_detuning}) enables us to study the graph-embedded Hamiltonian in Eq.~(\ref{eq:Rydberg_Ham}) both with and without the node-feature (local detuning) embedding, allowing for a direct, controlled comparison between the two settings.

We define the respective Hamiltonians as the (global) Hamiltonian 
\begin{align}
\hat{\mathcal{H}}_{\mathrm{glob}}^{\mathcal{G}} & =\hat{\mathcal{H}}\left(\mathbb{J}^{\mathcal{G}},\Omega,\boldsymbol{\delta}_{\mathrm{glob}}\right),\label{eq:H_global}
\end{align}
which is subjected to a globally fixed detuning value $\boldsymbol{\delta}_{\mathrm{glob}}=\left(0,\dots,0\right)$,
and as the (local) Hamiltonian
\begin{align}
\hat{\mathcal{H}}_{\mathrm{loc}}^{\mathcal{G}} & =\hat{\mathcal{H}}\left(\mathbb{J}^{\mathcal{G}},\Omega,\boldsymbol{\delta}^{\mathcal{G}}\right),\label{eq:H_local}
\end{align}
whose detuning $\boldsymbol{\delta}^{\mathcal{G}}=(\delta_{1}^{\mathcal{G}},\dots,\delta_{N}^{\mathcal{G}})$
embeds the node features via Eq.~(\ref{eq:local_detuning}). 
It is clear that a time evolution governed by a graph-dependent Hamiltonian,
Eq.~(\ref{eq:H_global}) or (\ref{eq:H_local}), leave an imprint
of the graph on the time-evolved state via Eq.~(\ref{eq:psi_t}),
and as a consequence, the measured observables (\ref{eq:observables})
equally acquire a graph dependence. This explains how the measured
observables encode a graph embedding, and why they can be meaningfully
used in a graph kernel application, \textit{cf.} Eq.~(\ref{eq:general_quantum_evolution_kernel}). 

The value of the fixed detuning $\delta_{0}$ needs to be chosen appropriately such that the impact of the local-detuning embedding used in Eq.~(\ref{eq:H_local}) is measurable. 
To estimate the magnitude we exploit the Ising representation of the Hamiltonian (\ref{eq:Ising_Ham}) wherein the longitudinal field value $h_{i}^{z}$ (\ref{eq:h_z}) directly demonstrates that the local detuning field has to compete with a relatively strong local offset value. 
For the given molecules, one can approximate the offset value as $\sum_{j\neq i}J_{ij}^{\mathcal{G}}\approx N_{i}^{\mathrm{deg}}J_{\mathrm{NN}}$ with $N_{i}^{\mathrm{deg}}$ being the degree of node $i$ and for molecules $N_{i}^{\mathrm{deg}}\in\left\{ 1,2,3,4\right\} $. 
Hence, the longitudinal field becomes 
\begin{align}
h_{i}^{z} & \approx\frac{1}{4\hbar}J_{\mathrm{NN}}\left[-N_{i}^{\mathrm{deg}}+\frac{2\hbar\delta_{0}}{J_{\mathrm{NN}}}\,\left(1-\frac{m_{C}}{m_{i}}\right)\right].\label{eq:h_z_approx}
\end{align}
Here, we note two things. 
First, in order for the local-detuning to
be noticeable, one has to choose the offset value such that $2\hbar\delta_{0}/J_{\mathrm{NN}} \sim \mathcal{O}\left(1\right)$.
Second, we see that the offset value conveniently encodes the degree of each node as a default piece of information.
We would not want to overshadow this embedded node information, but instead add a little value on top according to the node feature, \textit{i.e.} according to the mass of atom $i$. 
For this reason, we consider an appropriate choice for the fixed detuning value as $2\hbar\delta_{0}/J_{\mathrm{NN}}\lesssim0.3$.
The value we choose in this work is $\delta_{0}=2\pi\,\mathrm{rad}\,/\mathrm{\mu s}$, such that $2\hbar\delta_{0}/J_{\mathrm{NN}}\approx 0.23$.

\subsection{Emulation\label{subsec:emulation}}

With all the specifications provided we run two emulations for every graph in the datasets MUTAG and PTC\_FM$^*$: Once where the time evolution is governed by the local-detuning Hamiltonian $\hat{\mathcal{H}}_{\mathrm{loc}}^{\mathcal{G}}$ (\ref{eq:H_local}) and once governed by the global-detuning counterpart $\hat{\mathcal{H}}_{\mathrm{glob}}^{\mathcal{G}}$ (\ref{eq:H_global}). 
At multiple time steps, we sample the three observables defined in Eq.~(\ref{eq:observables}), the probability for a fixed number of excitations $P_k$, the excitation number $n_i$ and the correlation $C_{ij}$, with a number of shots $N_{\mathrm{shots}} = 1000$.
Their values form the set of feature vectors $\boldsymbol{O}^{\mathcal{G}_{\mu}}$, see Eq.~(\ref{eq:general_quantum_evolution_kernel}), which are used in the machine learning tasks in Secs.~\ref{sec:expe} and \ref{sec:kernels}. 
The emulation is computed using two open-source libraries: Pulser~\cite{Silverio2022pulseropensource} and EMU-MPS~\cite{emulators_pasqal}.
Pulser is a Python toolkit for neutral-atom quantum computing that allows to program experiments directly at the pulse level. The package EMU-MPS provides an efficient, state-of-the-art emulation backend that emulates pulse-level dynamics of neutral-atom arrays by storing the state as a matrix product state.

In Fig.~\ref{fig:dynamics}, we illustrate the time evolution by showing emulation results for a representative molecule from the PTC\_FM$^*$ dataset [panel (a)]. 
Panels (b)-(d) report, respectively, the excitation probability $P_k$, the site occupations $n_i$, and the correlation matrix elements $C_{i,j}=\langle \hat{n}_i \hat{n}_j \rangle$, cf. Eq.~(\ref{eq:observables}). 
The main purpose of this figure is to highlight the effect of the local detuning. 
To this end, we compare results obtained with global detuning (solid lines) to those obtained with local detuning (dashed lines) at $\delta_{0}=2\pi\,\mathrm{rad}\,/\mathrm{\mu s}$, cf. Eq.~(\ref{eq:local_detuning}). 
Across all three observables, deviations between the two settings become visible after $t \gtrsim 250 \, \mathrm{ns}$.
The molecule shown in panel (a) contains five carbon atoms and two oxygen atoms. 
Under the global-detuning Hamiltonian, Eq.~(\ref{eq:H_global}), all atoms effectively behave as carbon atoms ($m_\mathrm{C} = 12\,\mathrm{u}$). 
Consequently, the oxygen nodes ($0$ and $6$) are particularly informative. 
Consistent with this, Fig.~\ref{fig:dynamics}(c) shows that the carbon site $n_2$, which is distant from the oxygen atoms, is largely unaffected, whereas the oxygen sites ($n_0$ and $n_6$) and nearby carbon sites (e.g., $n_5$) exhibit pronounced differences between global and local detuning. 
Finally, to emphasize the role of the detuning scale $\delta_{0}$, the inset of Fig.~\ref{fig:dynamics}(d) shows the local-detuning correlation matrix results for a smaller value, $\delta_{0}=\pi\,\mathrm{rad}\,/\mathrm{\mu s}$. 
In this regime, global and local-detuning observables differ only marginally, so the resulting signal is likely too weak for a machine learning model to leverage.
As mentioned before, in this work we set the scale to $\delta_{0}=2\pi\,\mathrm{rad}\,/\mathrm{\mu s}$.

\subsection{Time-dependent perturbation theory\label{subsec:perturbation_theory}}

In this part we conduct an early-time expansion for the relevant observables (\ref{eq:observables}).
Specifically, we treat the transverse-field part of the Hamiltonian (\ref{eq:Rydberg_Ham}) as a perturbation, i.e. we assume $\Omega t\ll 1$.
The resulting early-time approximations allow us to gain useful insights into the performance of the quantum kernels introduced in Sec.~\ref{sec:kernels}.
Details on the derivation of the perturbation theory can be found in App.~\ref{app:perturbation theory}.

Neglecting all terms of order $\mathcal{O}\left( (\Omega t)^4 \right)$ and higher, the leading observable contributions are given as
\begin{align}
P_{0}\left(t\right) & = 1 - \frac{\Omega^2 t^2}{4} \sum_{i=1}^N \mathrm{sinc}^2\! \left( \delta_i t/2 \right) + \mathcal{O}\left( (\Omega t)^4 \right),\label{eq:pert_obs_P0}\\
P_{1}\left(t\right) & = \frac{\Omega^2 t^2}{4} \sum_{i=1}^N \mathrm{sinc}^2\! \left( \delta_i t/2 \right) + \mathcal{O}\left( (\Omega t)^4 \right),\label{eq:pert_obs_P1}\\
n_{i}\left(t\right) & = \frac{\Omega^2 t^2}{4}  \mathrm{sinc}^2\! \left( \delta_i t/2 \right) + \mathcal{O}\left( (\Omega t)^4 \right),\label{eq:pert_obs_ni}
\end{align}
where $\mathrm{sinc}\left(x \right) = \sin \left(x \right)/x$.
As shown in App.~\ref{app:perturbation theory} the excitation probability scales as $P_k\left(t\right) = \mathcal{O}\left( (\Omega t)^{2k} \right)$, such that $P_{k\geq 2}\left(t\right)$ can be neglected within the current approximation.
The parabolic expressions for $P_0\left( t\right)$ (\ref{eq:pert_obs_P0}) and $P_1\left( t\right)$ (\ref{eq:pert_obs_P1}) are consistent with Fig.~\ref{fig:dynamics}(b).
The diagonal elements of the correlation matrix $C_{ii}\left(t\right)=n_{i}\left(t\right)$ are approximated via Eq.~(\ref{eq:pert_obs_ni}), whereas the off-diagonal elements scale as $\mathcal{O}\left( (\Omega t)^{4} \right)$ consistent with the observed behavior in Fig.~\ref{fig:dynamics}(d). 
For early-time dynamics, where $\delta_i t\ll 1$, one can further approximate $\mathrm{sinc}\! \left( \delta_i t/2 \right)=1$, such that all site occupations $n_i\left(t \ll \delta_i^{-1} \right)$ behave identically, as can be seen in Fig.~\ref{fig:dynamics}(c).

\section{Graph Kernel methods for machine learning\label{sec:kernels}}

In this section we introduce all the kernels used in this work: 
first, the classical kernels used for benchmarking (\ref{subsec:classical_kernels}); then, the quantum feature kernels, $\mathrm{QEK}$ (\ref{subsec:kernel_QEK}) and $\mathrm{GDQC}$ (\ref{subsec:algo}). 
Lastly, in Subsec.~\ref{subsec:pooling} we introduce the pooling operations that are employed.

Kernel-based machine learning methods rely on a rigorous mathematical framework grounded in the theory of Reproducing Kernel Hilbert Spaces (RKHS). 
The foundation is laid by the Moore-Aronszajn Theorem, which guarantees that every positive definite kernel defines a unique RKHS where the kernel acts as the reproducing kernel~\cite{aronszajn50reproducing} . 
Within this space, the Representer Theorem shows that the solution to many regularized risk minimization problems can be written as a finite linear combination of kernel functions evaluated at the training points, thus reducing infinite-dimensional optimization to a tractable finite-dimensional one~\cite{scholkopf}. 
Furthermore, Mercer's Theorem~\cite{mercer1909theorem} offers a spectral decomposition of continuous, symmetric, positive-definite kernels, interpreting them as inner products in high-dimensional feature spaces~\cite{Scholkopf2002}. 
These results together enable the development of powerful nonlinear algorithms, such as Support Vector Machines~\cite{Cortes_1995}, with solid theoretical guarantees and practical scalability, for various  machine learning tasks such as classification and regression.
In practical terms these kernel methods allow SVMs to classify non-linearly separable data by implicitly mapping input features into a higher-dimensional space where a linear separator can be found, without the need of an explicit computation of the mapping~\cite{Scholkopf2002}.

\subsection{Classical graph kernels\label{subsec:classical_kernels}}

Graph-structured data are challenging for standard kernel methods because graphs do not lie in a vector space with a canonical notion of coordinates or Euclidean geometry. 
As a result, directly applying generic kernels is often nontrivial, since valid kernels must satisfy symmetry and positive semi-definiteness, and graph similarity is typically defined through discrete, combinatorial comparisons (see Subsec.~\ref{subsec:algo} for a discussion of properties). 
This has motivated the development of graph kernels that measure similarity by comparing substructures such as walks, subtrees, and subgraph patterns. 
By defining an inner product in an implicit feature space that preserves relevant combinatorial information, graph kernels provide a principled and flexible approach to graph machine learning tasks. 
From the extensive literature on classical graph kernels, we employ the following kernels for benchmarking (practical performance details can be found in App.~\ref{app:classical_graph_kernel}):

\paragraph{Random walk kernel}
This kernel measures similarity between graphs by counting matching random walks in each graph, thereby capturing detailed structural patterns. 
Introduced in Ref.~\cite{gaertner03graphkernels}, it is highly expressive but can be computationally intensive and prone to tottering effects. 
In our complexity model, the Random Walk kernel scales as $\mathcal{O}(N^6)$ and is therefore limited to relatively small graphs in practice. 
In its standard form, the kernel is purely structural and does not natively incorporate node attributes without additional feature engineering.

\paragraph{Graphlet sampling}
Graphlet kernels measure similarity between graphs by comparing the frequency of small induced subgraphs, known as graphlets, present in each graph, up to a specified maximum size~\cite{Przulj_2007, Shervashidze_2009}. 
In our setting, we enumerate all graphlets up to the chosen size (rather than sampling), which yields exact graphlet frequency vectors and captures local topological features and structural roles of nodes. 
This improved fidelity comes at higher computational cost and can increase memory consumption as the number of distinct graphlets grows. 
Exhaustive enumeration of all $k$-node graphlets scales as $\mathcal{O}(N^k)$. 
Unless attributes are explicitly included in the graphlet types (for example, by typing graphlets by labels), this baseline primarily captures topology.

\paragraph{Shortest path kernel}
Introduced in Ref.~\cite{SP_kernels}, the Shortest-Path (SP) kernel is an efficient alternative to random walks. 
It captures both local and global structure by matching pairs of nodes with identical shortest-path lengths. 
Its computational complexity scales as $\mathcal{O}(N^4)$ in our setting. 
This kernel can incorporate node labels by encoding each node pair as $(\ell(u), \ell(v), D(u,v))$, which enriches the feature space beyond distances alone. 
However, this also increases the representation dimensionality: the feature size grows from $\mathcal{O}(d_{\mathrm{max}})$ to $\mathcal{O}(|L|^2\, d_{\mathrm{max}})$, where $|L|$ denotes the label vocabulary size and $d_{\mathrm{max}}$ the graph diameter.

\paragraph{$k$-core kernel}
This kernel measures similarity between graphs by comparing statistics (e.g., distributions or distances) extracted from their $k$-core decomposition. 
The $k$-core is the maximal induced subgraph in which every node has degree at least $k$, and the sequence of cores as $k$ increases provides a multiscale view of graph connectivity and cohesion. 
Following Ref.~\cite{k_core_kernels}, each graph is represented through feature vectors derived from its cores, and graph similarity is computed by applying a base kernel to these structured representations. 
In our experiments, we use the \emph{Shortest-Path} kernel as base kernel, thereby combining core-based structural information with distance-aware matching. 
The $k$-core decomposition can be computed in $\mathcal{O}(\delta_{\mathrm{min}}\, N + M)$ time, where $\delta_{\mathrm{min}} \ll N$ denotes the minimum degeneracy and $M$ the number of edges (the overall cost then includes the cost of the base kernel).

\paragraph{Weisfeiler-Lehman optimal assignment kernel}
Building on the Weisfeiler-Lehman (WL) color-refinement procedure, the WL optimal assignment kernel compares two graphs by deriving refined vertex labels across WL iterations and then computing an optimal one-to-one correspondence between their vertex sets under a WL-based base kernel~\cite{Kriege_2016}. 
Concretely, vertices are scored by the number of WL iterations in which their labels match, and the graph-level similarity is obtained by maximizing the total similarity over all bijections between vertices. 
The resulting assignment kernel can improve upon the WL subtree kernel while remaining valid. 
Its computational complexity scales as $\mathcal{O}(h\, M + N)$, where $h$ denotes the number of WL iterations. 
Importantly, the method naturally supports node labels by initializing the refinement with the provided attributes: with labels, nodes start from attribute-derived labels, yielding finer refinements; without labels, all nodes share the same initial label and the kernel captures topology only. 
Incorporating labels increases the size of the induced label space, which can lead to sparse, high-dimensional representations and weaker generalization on small datasets.

\subsection{Quantum evolution kernel ($\mathrm{QEK}$)\label{subsec:kernel_QEK}}

First introduced by Henry \textit{et al.}~\cite{Henry_2021}, the quantum evolution kernel algorithm is a quantum-feature kernel implementation [see Eq.~(\ref{eq:general_quantum_evolution_kernel})] that draws from probability distributions generated from measured observables.
Specifically, within the proof-of-principle implementation of Ref.~\cite{Albrecht2023}, the algorithm draws from the global observable vector $\boldsymbol{P}^{\mathcal{G}}=\left(P_0^{\mathcal{G}}(t), P_1^{\mathcal{G}}(t), P_2^{\mathcal{G}}(t), \dots \right)^\intercal$ with $P_k^{\mathcal{G}}(t)$ being the excitation probability for a fixed number $k$ of excitations at a given time $t$, see Eq.~(\ref{eq:observables}). 
The superscript $\mathcal{G}$ indicates that the time evolution is subjected to a graph-embedded Hamiltonian, \textit{cf.} Sec.~\ref{sec:Quantum-graph-encoding}.
To discriminate between two graphs, $\mathcal{G}_{\mu}$ and $\mathcal{G}_{\mu^\prime}$, the QEK algorithm implements the following kernel function 
\begin{align}
    \kappa_{\mathrm{QEK}}(\mathcal{G}_{\mu}, \mathcal{G}_{\mu^\prime}) = \exp\left[ -\mu_0 \,\mathrm{JS}\left(\boldsymbol{P}^{\mathcal{G_\mu}}, \boldsymbol{P}^{\mathcal{G_{\mu^\prime}}} \right) \right],\label{eq:QEK_kernel}
\end{align}
with the model parameter $\mu_0>0$ set to $\mu_0 = 2$ throughout this work.
The Jensen-Shannon divergence $\mathrm{JS}\left(\boldsymbol{P}, \boldsymbol{P}^\prime \right)$ measures the similarity between the two distributions.
By construction, the kernel function (\ref{eq:QEK_kernel}) ensures positive definiteness. 

$\mathrm{QEK}$ post-processing mainly involves turning the measurement results into an excitation probability distribution. 
With $N_{\mathrm{shots}}$ measurement shots, and a graph of size $N$, forming $\boldsymbol{P}^G$ takes $\mathcal{O}(N_{\mathrm{shots}}\,N)$ operations (one processes up to $N$ bits per shot). 
Computing the Jensen-Shannon divergence from this distribution then costs an additional $\mathcal{O}(N)$. 
Overall, the total post-processing complexity is therefore dominated by $\mathcal{O}(N_{\mathrm{shots}}\,N)$.

Within its proof-of-principle implementation~\cite{Albrecht2023}, the $\mathrm{QEK}$ kernel (\ref{eq:QEK_kernel}) has been applied to non-attributed graphs. 
Correspondingly, the time evolution has been subjected to a global-detuning Hamiltonian $\hat{\mathcal{H}}_{\mathrm{glob}}^{\mathcal{G}}$ (\ref{eq:H_global}).
When coupled with a support vector machine, the $\mathrm{QEK}$ algorithm has matched the performance of leading state-of-the-art classical graph kernels on several benchmark datasets of small graphs~\cite{Albrecht2023}.
Thereby, $\mathrm{QEK}$`s capabilities have been demonstrated both, numerically via emulations as well as on real quantum hardware, on a 32-qubit neutral-atom processor.

In our work, we expand on the previous $\mathrm{QEK}$ implementation~\cite{Albrecht2023} twofold. 
First, and most importantly, we want it to be able to deal with attributed graphs. 
This is accomplished by changing the Hamiltonian that governs the time evolution from $\hat{\mathcal{H}}_{\mathrm{glob}}^{\mathcal{G}}$ (\ref{eq:H_global}) to the local-detuning Hamiltonian $\hat{\mathcal{H}}_{\mathrm{loc}}^{\mathcal{G}}$ (\ref{eq:H_local}).
As described in Subsec.~\ref{subsec:Graph_Embedding}, the latter Hamiltonian embeds the atomic masses (the node features) into the local detuning $\delta_i$.
As a result, we obtain an observable vector $\boldsymbol{P}^{\mathcal{G}}$ which contains graph information drawn from both, the edges and the node features.
Similar to before, the observable vectors are inserted into the kernel function (\ref{eq:QEK_kernel}). 
To distinguish the two cases, evolution under $\hat{\mathcal{H}}_{\mathrm{glob}}^{\mathcal{G}}$ and $\hat{\mathcal{H}}_{\mathrm{loc}}^{\mathcal{G}}$, we use the notation $\kappa_{\mathrm{QEK}}^{\mathrm{glob}}$ and $\kappa_{\mathrm{QEK}}^{\mathrm{loc}}$, respectively.
Mathematically, switching from $\kappa_{\mathrm{QEK}}^{\mathrm{glob}}$ to $\kappa_{\mathrm{QEK}}^{\mathrm{loc}}$ has profound consequences which can be summarized by proposition \ref{prop:prop1} which is proven in App.~\ref{app:proof_prop1}. 
\begin{proposition}
\label{prop:prop1}
The quantum evolution kernel $\kappa_{\mathrm{QEK}}^{\mathrm{loc}}$
is genuinely more expressive than $\kappa_{\mathrm{QEK}}^{\mathrm{glob}}$ for distinguishing attributed graphs. 
\end{proposition}
The second extension to the $\mathrm{QEK}$ algorithm that we undertake lies in pooling of different time steps. 
The above implementation draws the excitation probability from a single time step $t$.
In Subsec.~\ref{subsec:pooling} we explain how multiple time steps can be combined to accrue information sampled from different regimes of the time evolution.

To gain intuition about the performance of $\mathrm{QEK}$ at later times, we also define a classical baseline. 
To this end, we exploit the early-time dynamics (see Subsec.~\ref{subsec:perturbation_theory}) and train the model using only leading-order approximation to the excitation-probability distribution,
\begin{align}
\boldsymbol{P}^{\mathcal{G}}_{\mathrm{baseline}} & =\left( 1 - N \Omega^2 t^2/4, N \Omega^2 t^2/4, 0, \dots ,0\right)^\intercal ,
\label{eq:P_vector_early_time_approx}
\end{align}
cf. Eqs.~(\ref{eq:pert_obs_P0})-(\ref{eq:pert_obs_P1}). 
Importantly, the only graph-dependent quantity entering Eq.~(\ref{eq:P_vector_early_time_approx}) is the graph size $N=\lvert\mathcal{V}\rvert$. 
This baseline therefore isolates the contribution of graph size and provides a convenient reference when evaluating $\mathrm{QEK}$ trained on the full quantum data, i.e., outside from the early-time regime.
When applying the baseline (\ref{eq:P_vector_early_time_approx}) to the datasets in Sec.~\ref{sec:expe}, we choose $\Omega t=0.25$, which satisfies $N \Omega^2 t^2/4<1$ for all graphs in the datasets.

The excitation probability $P_k^{\mathcal{G}}(t)$ (\ref{eq:observables}) used in the $\mathrm{QEK}$ algorithm is a natural observable candidate to capture graph information as it is permutation invariant. 
Moreover, it is a global observable rendering it a useful tool to study global graph properties. 
However, if it comes to local graph properties the excitation probability $P_k^{\mathcal{G}}(t)$ may not be the most convenient choice. 
An attempt to generalize the algorithm from the global observables $P_k^{\mathcal{G}}(t)$ to observables that capture distributions of local subcomponents of a graph (\textit{e.g.}, nodes or edges), while maintaining permutation invariance, requires the fixing of a canonical ordering for each graph, a costly task known as the graph canonization problem~\cite{gurevich1997from}. 
This potential shortcoming of the $\mathrm{QEK}$ algorithm (\ref{eq:QEK_kernel}) has motivated the design of the generalized-distance quantum-correlation kernel, see Subsec.~\ref{subsec:algo}.

\subsection{Generalized-distance quantum-correlation ($\mathrm{GDQC}$) kernel\label{subsec:algo}}

In this section, we introduce the generalized-distance quantum-correlation kernel, a new quantum-feature kernel used in this work. 
Unlike $\mathrm{QEK}$, which is designed to capture global graph information, $\mathrm{GDQC}$ is based on local observables and is therefore better suited to predicting local graph properties.

The $\mathrm{GDQC}$ kernel is inspired by the classical generalized-distance Weisfeiler-Leman (GD-WL) kernel~\cite{shervashidze11a}. Following this idea, we combine a graph distance with an additional source of graph information. We denote by $d_{\mathcal{G}}(i,j)$ the distance between nodes $i$ and $j$, and by $D^{\mathcal{G}}$ the corresponding distance matrix with entries $D^{\mathcal{G}}_{i,j} = d_{\mathcal{G}}(i,j)$. 
In this work, we use the shortest-path distance as the standard metric for characterizing pairwise node distance. This choice offers two key advantages: first, on unweighted graphs it takes integer values, which provides a natural discretization for one of the model's hyper-parameters and simplifies parameterization. Second, as one of the most well-established graph-theoretic metrics, it benefits from a rich ecosystem of efficient algorithms, making it both principled and computationally practical.
\newline Within the $\mathrm{GDQC}$ framework, the second element is a local observable. Concretely, we use the correlation matrix $C^{\mathcal{G}}(t)$ with entries $C^{\mathcal{G}}_{i,j}(t) = \left\langle \psi_{\mathcal{G}}(t) \big| \hat{n}_i \hat{n}_j \big| \psi_{\mathcal{G}}(t) \right\rangle$, cf. Eq.~(\ref{eq:observables}). More generally, the $\mathrm{GDQC}$ framework can be extended to other local observables, including higher-order correlators.

\begin{algorithm}[t]
\DontPrintSemicolon
\caption{Computation of the $\mathrm{GDQC}$ feature vector $\bm{\chi}\!\left[C^{\mathcal{G}}\!\left(t\right),D^{\mathcal{G}}\right]$ in Eq.~(\ref{eq:GDQC_t}).}
\label{algo:GDQC}
\KwIn{$C^{\mathcal{G}}(t)$; $D^{\mathcal{G}}$; $N_{\mathrm{bins}}^D$; $N_{\mathrm{bins}}^C$; $\Delta_C = 1/N_{\mathrm{bins}}^C$; $\Delta_D = d_\mathrm{max}/N_{\mathrm{bins}}^D$}
\For{$\ell \leftarrow 0$ \KwTo $N_{\mathrm{bins}}^D - 1$}{
  $\Omega_\ell^{(d)} \gets \bigl[\,\ell\,\Delta_D,\, (\ell+1)\,\Delta_D \bigr]$\;
  \For{$p \leftarrow 0$ \KwTo $N_{\mathrm{bins}}^C - 1$}{
    $\Omega_p^{(q)} \gets \bigl[\,p\,\Delta_C,\, (p+1)\,\Delta_C \bigr]$\;
    $\bm{\chi}\bigl[N_{\mathrm{bins}}^C\,\ell + p \bigr]
      \gets \Bigl|\{ (i,j) : D^{\mathcal{G}}_{i,j} \in \Omega_\ell^{(d)},\ 
      C^{\mathcal{G}}_{i,j}(t) \in \Omega_p^{(q)} \} \Bigr|$\;
  }
}
\Return{$\bm{\chi}\, / \lVert \bm{\chi} \rVert$}\;
\end{algorithm}

The proposed $\mathrm{GDQC}$ kernel function [\textit{cf.} Eq.~(\ref{eq:general_quantum_evolution_kernel})] between two graphs, $\mathcal{G}_{\mu}$ and $\mathcal{G}_{\mu^\prime}$, is defined as scalar product,
\begin{align}
 \kappa_{\mathrm{GDQC}}\left(\mathcal{G}_{\mu}, \mathcal{G}_{\mu^\prime}\right)
& =\boldsymbol{\chi}^{\intercal}\left[C^{\mathcal{G}_{\mu}}\left(t\right),D^{\mathcal{G}_{\mu}}\right]\,\boldsymbol{\chi}\left[C^{\mathcal{G}_{\mu^{\prime}}}\left(t\right),D^{\mathcal{G}_{\mu^{\prime}}}\right].
\label{eq:GDQC_t}
\end{align}
It involves a graph-dependent $\mathrm{GDQC}$ feature vector $\boldsymbol{\chi}\left[C^{\mathcal{G}}\left(t\right),D^{\mathcal{G}}\right]$ whose construction is the core element of the $\mathrm{GDQC}$ framework.
Mathematically, it is defined via algorithm~\ref{algo:GDQC}, and its construction is illustrated in Fig.~\ref{fig:GDQC_illustration}. 
It relies on a binning procedure where the $N^2$ node pairs $\left(i,j \right)$ are binned in two steps: 
first, according to the value of their correlation-matrix elements $C^{\mathcal{G}}_{i,j}\left(t\right) \in \left[0,1 \right]$, and second, according to their distance $d_{\mathcal{G}}(i,j)$.
Hereby, the correlation values are binned into $N_{\mathrm{bins}}^C \geq 1$ intervals of size $\varDelta_C = 1/N_{\mathrm{bins}}^C$ (bins associated with the upper axis in Fig.~\ref{fig:GDQC_illustration}), while the distances are binned into $N_{\mathrm{bins}}^D \geq 1$ intervals of size $\varDelta_D = d_\mathrm{max}/N_{\mathrm{bins}}^D$ where $d_\mathrm{max}$ denotes the maximum distance across the entire dataset.
Using the integer-valued shortest-path distance (or graph distance), the most natural binning choice is
$N_{\mathrm{bins}}^D=d_\mathrm{max}+1$ bins, \textit{i.e.} one bin for each possible graph distance (including distance $0$) as illustrated on the lower axis in Fig.~\ref{fig:GDQC_illustration}.
For the correlation values, we instead treat the binning resolution as a hyperparameter: we analyze the effect of the correlation-bin width $\varDelta_C = 1/N_{\mathrm{bins}}^C$ in Subsec.~\ref{subsec:binning_analysis}. 
Unless stated otherwise, the results in the main part of Sec.~\ref{sec:expe} use $N_{\mathrm{bins}}^C = 10$.
In combination the two binning procedures create the $(N_{\mathrm{bins}}^C N_{\mathrm{bins}}^D)$-dimensional $\mathrm{GDQC}$ feature vector $\boldsymbol{\chi}$ where each element is associated with a specific combination of correlation and distance bins.

Computing the correlation matrix from $N_{\mathrm{shots}}$ measurement outcomes requires $\mathcal{O}(N_{\mathrm{shots}}\,N^2)$ operations. 
Computing the all-pairs shortest-path distance matrix with the Floyd-Warshall algorithm~\cite{FloydWarshall} costs $\mathcal{O}(N^3)$. 
Finally, the binning step is linear in the number of node pairs, and therefore costs $\mathcal{O}(N^2)$.
Overall, the $\mathrm{GDQC}$ post-processing complexity is $\mathcal{O}(N_\mathrm{shots} \, N^2 + N^3)$.
In our experiments, $N_{\mathrm{shots}} > N$, so the dominant term is $\mathcal{O}(N_{\mathrm{shots}}\,N^2)$.

Following the protocol used for $\mathrm{QEK}$ (Subsec.~\ref{subsec:kernel_QEK}), we evaluate the kernel in Eq.~(\ref{eq:GDQC_t}) under two time evolutions: one generated by the global-detuning Hamiltonian $\hat{\mathcal{H}}_{\mathrm{glob}}^{\mathcal{G}}$ (\ref{eq:H_global}) and one generated by the local-detuning Hamiltonian $\hat{\mathcal{H}}_{\mathrm{loc}}^{\mathcal{G}}$ (\ref{eq:H_local}). 
We denote the corresponding kernels by $\kappa_{\mathrm{GDQC}}^{\mathrm{glob}}$ and $\kappa_{\mathrm{GDQC}}^{\mathrm{loc}}$. 
In our experiments, these are used for non-attributed and attributed graphs, respectively.
\begin{figure}[t]
\includegraphics[width=1\columnwidth]{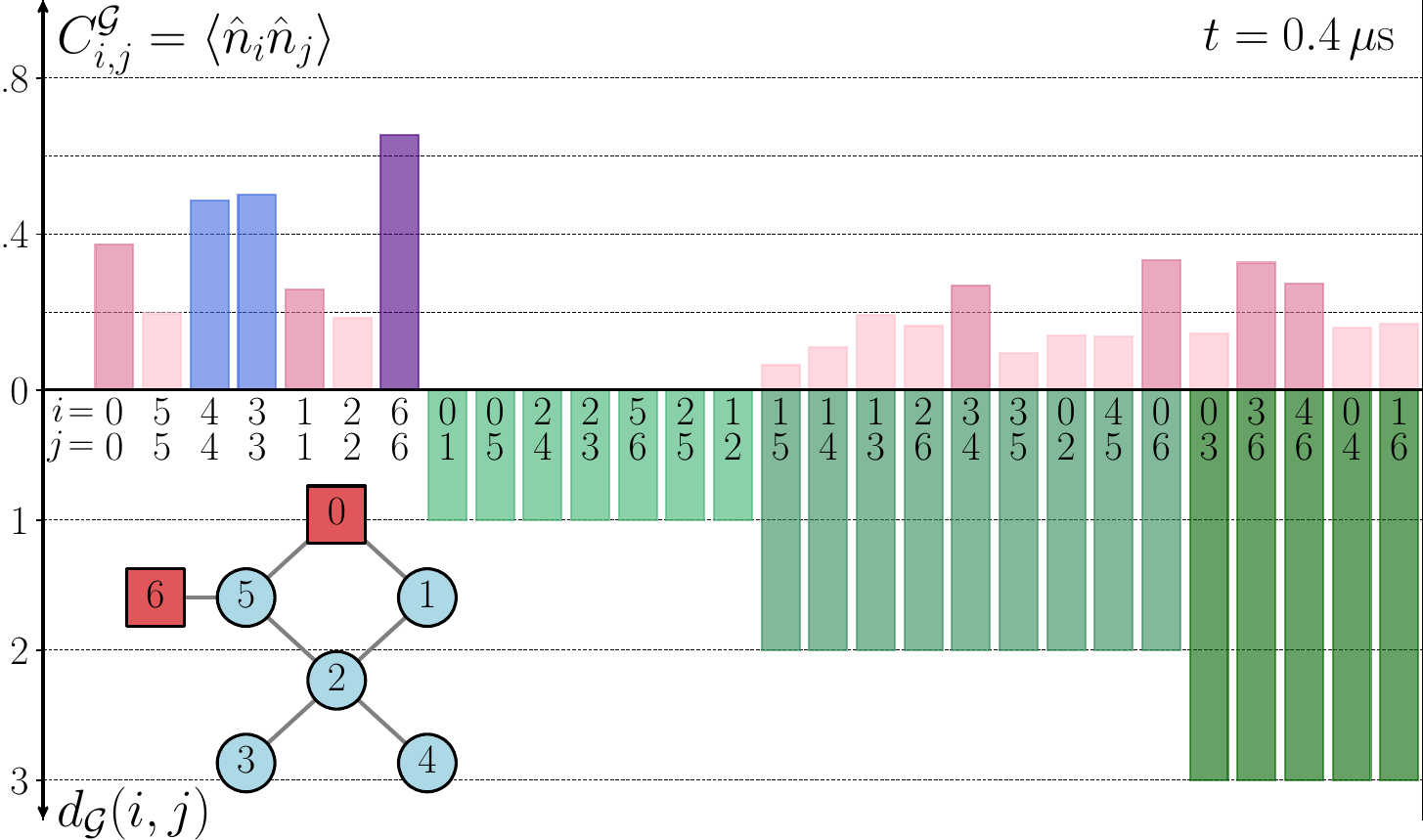}
\caption{Illustration for the construction of the generalized-distance quantum-correlation feature vector $\boldsymbol{\chi}\left[C^{\mathcal{G}}\left(t\right),D^{\mathcal{G}}\right]$ according to the algorithm~\ref{algo:GDQC} for the indicated graph of size $N=7$. 
The figure shows the measured correlation $C^{\mathcal{G}}_{i,j}$ (upper bar) together with the respective distance measure $d_{\mathcal{G}}(i,j)$ (lower bar) for each of the $N^2=49$ node pairs $\left(i,j \right)$.
Note the symmetries $C^{\mathcal{G}}_{j,i}=C^{\mathcal{G}}_{i,j}$ and $d_{\mathcal{G}}(j,i)=d_{\mathcal{G}}(i,j)$.
Also, notice how the pairs that are at distance 1 showcase very low values due to the blockade phenomenon.
Having chosen $N_{\mathrm{bins}}^C = 5$ correlation bins, the node pairs $\left(i,j \right)$ are binned according to their correlation-matrix values (indicated by the coloring), leading to an intermediate correlation-matrix binning vector $\boldsymbol{\chi}^C_{\mathrm{bin}} =\left( 36, 10, 2, 1, 0\right)^\intercal$.
Next, the node pairs are further binned according to the graph distance $d_{\mathcal{G}}(i,j)$. 
For example, the 36 elements from the interval $[0, 0.2]$ (faint red values) are redistributed according to the four involved graph distance values, leading to a sequence $\{2, 14, 14, 6\}$.
All bins combined, this procedure generates the $(N_{\mathrm{bins}}^C N_{\mathrm{bins}}^D)$-dimensional $\mathrm{GDQC}$ feature vector $\boldsymbol{\chi}$ which, in the present example becomes $\boldsymbol{\chi} =
(2, 2, 2, 1, 0, 14, 0, 0, 0, 0, 14, 4, 0, 0, 0, 6, 4, 0, 0, 0)^\intercal /\sqrt{473}$.
The presented correlation matrix $C^{\mathcal{G}}_{i,j}=C^{\mathcal{G}}_{i,j}\left(t = 0.4\, \mathrm{\mu s}\right)$ corresponds to the time step indicated in Fig.~\ref{fig:dynamics}.
}
\label{fig:GDQC_illustration}
\end{figure}

\begin{proposition}\label{prop:prop2}
The function $\kappa_{\mathrm{GDQC}}$ defined in Eq.~(\ref{eq:GDQC_t}) is a valid graph kernel, i.e. it is symmetric positive semi-definite on the space of graphs equipped with any graph distance $d_{\mathcal{G}}$, and is isomorphism-invariant such that $\kappa_{\mathrm{GDQC}}(\mathcal{G}_{\mu}, \mathcal{G}_{\mu^\prime}) = 1$ whenever $\mathcal{G}_\mu$ and $\mathcal{G}_{\mu^\prime}$ are isomorphic.
\end{proposition} 
Next, we demonstrate that the $\mathrm{GDQC}$ kernel function $\kappa_{\mathrm{GDQC}}$ (\ref{eq:GDQC_t}) satisfies all graph-kernel properties. 
Its key properties are summarized in proposition~\ref{prop:prop2}, and are proven in detail in App.~\ref{app:GDQC}. 
Let us emphasize that the converse of point (3) is not proven to be true, \textit{that is,} we do not claim that $\kappa_{\mathrm{GDQC}}(\mathcal{G}_\mu,\mathcal{G}_{\mu^\prime}) = 1$ if \textit{and only if} $\mathcal{G}_\mu$ and $\mathcal{G}_{\mu^\prime}$ are isomorphic. A proof thereof (if it was true) would require a deeper analysis. 

The connection between the proposed $\mathrm{GDQC}$ kernel~(\ref{eq:GDQC_t}) and the Weisfeiler-Lehman (WL) family of kernels~\cite{shervashidze11a}, in particular its generalized-distance variant (GD-WL), can be made explicit by comparing their \emph{expressiveness}. 
Let $\kappa_1$ and $\kappa_2$ be two normalized graph kernels.
\newpage
\begin{enumerate}
\item A kernel $\kappa$ is \emph{expressive} for a pair of non-isomorphic graphs $(\mathcal{G},\mathcal{G}^\prime)$ if it distinguishes them, i.e.,
$\kappa(\mathcal{G},\mathcal{G}^\prime)<1$.
\item We say that $\kappa_1$ is \emph{more expressive} than $\kappa_2$ if
(i) for all non-isomorphic $(\mathcal{G},\mathcal{G}^\prime)$, $\kappa_1(\mathcal{G},\mathcal{G}^\prime)=1 \Rightarrow \kappa_2(\mathcal{G},\mathcal{G}^\prime)=1$, and
(ii) there exists at least one non-isomorphic pair $(\mathcal{G}_*,\mathcal{G}^\prime_*)$ such that $\kappa_2(\mathcal{G}_*,\mathcal{G}^\prime_*)=1$ but $\kappa_1(\mathcal{G}_*,\mathcal{G}^\prime_*)<1$.
\end{enumerate}

A practical question then concerns the binning parameters $N_{\mathrm{bins}}^D$ and $N_{\mathrm{bins}}^C$, which determine the bin widths $\Delta_D$ and $\Delta_C$. 
Increasing the number of bins refines the induced features and typically increases expressiveness, as formalized by Proposition~\ref{prop:prop3} whose proof is provided in App.~\ref{app:expressivity_gdqc}.

\begin{proposition}\label{prop:prop3}
The $\mathrm{GDQC}$ kernel $\kappa_{\mathrm{GDQC}}$ matches the expressiveness of the GD-WL refinement~\cite{zhang2023rethinking} when GD-WL is initialized with node colors constructed from the multisets of (binned) correlation values in the rows/columns of $C^{\mathcal{G}}(t)$.
\end{proposition}

If the chosen graph distance induces a GD-WL refinement that is less expressive than the quantum correlation matrix, then the expressiveness of $\kappa_{\mathrm{GDQC}}$ cannot exceed that of the correlation matrix, since the latter already provides a stable initial coloring for GD-WL (see Appendix~\ref{app:expressivity_gdqc}). This point is crucial: Zhang et al.~\cite{Zhang_expressive} show that, for several commonly used distances (including Laplacian-eigenvector, resistance, and shortest-path distances), the expressiveness of GD-WL is upper bounded by $3$-WL. In contrast, Thabet et al.~\cite{Thabet_Quantum} provide evidence that correlation matrices of local observables generated by an Ising Hamiltonian can be at least as powerful as $4$-WL, as demonstrated by their ability to distinguish strongly regular graphs~\cite{godsil01,pmlr-v139-bodnar21a}. Taken together, these results suggest that, for the state-of-the-art distances studied in~\cite{Zhang_expressive}, using a ``simple'' distance (e.g., shortest-path distance) is unlikely to increase the expressiveness of $\kappa_{\mathrm{GDQC}}$ beyond what is already captured by the quantum correlations, although it may still improve downstream predictive performance (cf. Sec.~\ref{sec:expe}).

Finally, we provide intuition for the $\mathrm{GDQC}$ kernel by identifying a purely classical baseline. 
To this end, we exploit the early-time dynamics discussed in Subsec.~\ref{subsec:perturbation_theory}. 
We have shown that the diagonal correlation-matrix elements scale as
$C_{i,i}\left(t\right)=\mathcal{O}\left( (\Omega t)^{2} \right)$,
whereas the off-diagonal elements scale as
$C_{i,j\neq i}\left(t\right)=\mathcal{O}\left( (\Omega t)^{4} \right)$ for $i\neq j$.
Consequently, at sufficiently early times all entries of $C(t)$ are small. 
The $\mathrm{GDQC}$ correlation binning therefore assigns all $N^2$ node pairs to the first correlation bin, while the remaining bins are empty.
Hence, the intermediate correlation-matrix binning vector becomes
\begin{align}
    \boldsymbol{\chi}^C_{\mathrm{bin},\, \mathrm{baseline}} &=\left( N^2, 0, \dots , 0\right)^\intercal,\label{eq:chi_bin_classical}
\end{align}
cf. Fig.~\ref{fig:GDQC_illustration}. 
In the near-zero correlation regime, the correlation-binning step carries no discriminative information, so the subsequent distance binning becomes the only nontrivial component.
Consequently, the $\mathrm{GDQC}$ feature vector $\boldsymbol{\chi}$ is dominated by the graph's distance profile: for each shortest-path length, it counts (up to normalization) how many node pairs lie at that distance, making its early-time behavior closely related to the classical (unattributed) shortest-path kernel (see Subsec.~\ref{subsec:classical_kernels}). 
Notably, the representation is not strictly identical, since $\mathrm{GDQC}$ also includes contributions from node pairs at zero graph distance, whereas the standard shortest-path kernel excludes diagonal terms.
Overall, the early-time $\mathrm{GDQC}$ regime provides a natural classical baseline.

\subsection{Kernel pooling\label{subsec:pooling}}

We now describe how to combine information from observables sampled at multiple times $\{t_1,\dots,t_{N_t}\}$ (cf. Eq.~\ref{eq:general_quantum_evolution_kernel}), while still keeping relevant input for the task at hand. 
Two strategies are possible: (i) concatenate all time-dependent features into a single vector and apply one kernel, or (ii) define one kernel per time step and pool these kernels. 
We follow the second strategy and construct a collection of kernels $\{\kappa_1,\dots,\kappa_{N_t}\}$, where $\kappa_\gamma$ is computed from observables at time $t_\gamma$.

Any pooling rule that preserves kernel properties [\textit{cf.} proposition~\ref{prop:prop2}] yields a valid kernel. In particular, if each $\kappa_\gamma$ is a valid kernel, then (a) any convex combination and (b) the pointwise product are also valid kernels. In this work, we consider the following two pooling schemes:
\begin{align}
\kappa_{+}(\mathcal{G}_1,\mathcal{G}_2)
&= \sum_{\gamma=1}^{N_t} \alpha_\gamma\, \kappa_\gamma(\mathcal{G}_1,\mathcal{G}_2),
\label{eq:lin_comb}\\
\kappa_{\times}(\mathcal{G}_1,\mathcal{G}_2)
&= \prod_{\gamma=1}^{N_t} \kappa_\gamma(\mathcal{G}_1,\mathcal{G}_2),
\label{eq:prod}
\end{align}
where $\alpha_\gamma \ge 0$ and $\sum_{\gamma=1}^{N_t}\alpha_\gamma = 1$.
From a feature-map perspective, if $\kappa_\gamma(\mathcal{G}_1,\mathcal{G}_2)=\langle \boldsymbol{\chi}_\gamma(\mathcal{G}_1),\boldsymbol{\chi}_\gamma(\mathcal{G}_2)\rangle$, then Eq.~\eqref{eq:lin_comb} corresponds to the concatenated map
\begin{equation}
\boldsymbol{\chi}_{+}(\mathcal{G}) = \big\|_{\gamma=1}^{N_t} \left(\sqrt{\alpha_\gamma}\,\boldsymbol{\chi}_\gamma(\mathcal{G})\right),
\end{equation}
where $||_{\gamma = 1}^{N_t}$ denotes concatenation, while Eq.~\eqref{eq:prod} corresponds to the tensor-product map
\begin{equation}
\boldsymbol{\chi}_{\times}(\mathcal{G}) = \bigotimes_{\gamma=1}^{N_t} \boldsymbol{\chi}_\gamma(\mathcal{G}),
\end{equation}
where $\otimes$ denotes the tensor product. These constructions suggest that pooling \emph{can} improve expressiveness, but this is not guaranteed and depends on the additional kernels. We formalize this intuition in Proposition~\ref{prop:prop4} (see App.~\ref{app:kernel_pooling} for a proof).
\begin{proposition}\label{prop:prop4}
Let $S=\{\kappa_1,\dots,\kappa_N\}$ be a set of graph kernels, and let $\kappa_A$ and $\kappa_B$ be kernels obtained by pooling two subsets $S_A$ and $S_B$ of $S$ using (\ref{eq:lin_comb}) or (\ref{eq:prod}) as pooling rules. If $S_A \subset S_B$, then $\kappa_B$ is at least as expressive as $\kappa_A$.
\end{proposition}

\begin{table}[t]
\centering
\caption{Weighted F1 scores associated with the best-performing SVM models on two datasets, MUTAG and PTC\_FM$^*$, for various kernel implementations introduced in Sec.~\ref{sec:kernels}: ranging from the classical kernels to the two quantum-feature kernels, $\mathrm{QEK}$ (\ref{eq:QEK_kernel}) and $\mathrm{GDQC}$ (\ref{eq:GDQC_t}).
The superscript $\mathrm{(attr)}$ on classical kernels denotes the incorporation of node features (attributes).
The quantum-feature kernels are listed in a non-pooled and a pooled version where the subscript $\boldsymbol{+}$ ($\boldsymbol{\times}$) indicates sum- (Hadamard product-) pooling, \textit{cf.} Subsec.~\ref{subsec:pooling}. 
The superscripts $\mathrm{(loc)}$ and $\mathrm{(glob)}$ refer to whether the governing Hamiltonian has node features embedded [$\hat{\mathcal{H}}_{\mathrm{loc}}^{\mathcal{G}}$ (\ref{eq:H_local})] or not [$\hat{\mathcal{H}}_{\mathrm{glob}}^{\mathcal{G}}$ (\ref{eq:H_global})].
The \textcolor{JungleGreen}{best-}, \textcolor{BurntOrange}{second-to-best}, and \textcolor{Cerulean}{third-to-best} performing models are highlighted in the respective colors.
}

\begin{tabular}{llcc}
\toprule
 & \textbf{Kernel} & \textbf{MUTAG} & \textbf{PTC\_FM$^*$} \\
& & \textbf{F1 score (\%)} & \textbf{F1 score (\%)}\\
\midrule
\multirow{8}{*}{\rotatebox{90}{Classical}} 
& Random Walk                     & 85.60 $\pm$ 4.23     & 58.12 $\pm$ 4.76   \\
& Graphlet sampling               & 84.76 $\pm$ 5.41     & 60.69 $\pm$ 5.11     \\
& $k$-core                        & 83.58 $\pm$ 6.00     & 58.98 $\pm$ 4.66     \\
& Shortest path                   & 82.25 $\pm$ 5.69     & 54.06 $\pm$ 5.12    \\
& WL optimal assignment           & 86.88 $\pm$ 5.01     & 58.07 $\pm$ 5.00     \\
& $\text{$k$-core}^{\mathrm{(attr)}}$                & 83.99 $\pm$ 5.39     & 59.84$\pm$ 4.76     \\
& $\text{Shortest path}^{\mathrm{(attr)}}$           & 81.70 $\pm$ 5.40     & 59.71 $\pm$ 4.63     \\
& $\text{WL optimal assignment}^{\mathrm{(attr)}}$   & 86.63 $\pm$ 5.25    & 63.40 $\pm$ 4.23 \\
\midrule
\multirow{12}{*}{\rotatebox{90}{Quantum features}}
& $\mathrm{QEK}^{\mathrm{(loc)}}$                      & 86.19 $\pm$ 5.03     & 64.11 $\pm$ 4.75    \\
& $\mathrm{QEK}^{\mathrm{(glob)}}$                     & 85.80 $\pm$ 5.81     & 60.62 $\pm$ 5.16    \\
& $\mathrm{GDQC}^{\mathrm{(loc)}}$                     & 89.29 $\pm$ 4.89     & 60.87 $\pm$ 5.49    \\
& $\mathrm{GDQC}^{\mathrm{(glob)}}$                    & 88.53 $\pm$ 4.33     & 59.61 $\pm$ 5.57    \\
& $\mathrm{QEK}^{\mathrm{(loc)}}_{\boldsymbol{+}}$      & 89.26 $\pm$ 5.40     & \textbf{\textcolor{JungleGreen}{69.55 $\pm$ 4.98}}    \\
& $\mathrm{QEK}^{\mathrm{(loc)}}_{\boldsymbol{\times}}$ & 88.78 $\pm$ 4.89 & \textbf{\textcolor{BurntOrange}{69.28 $\pm$ 5.28}} \\
& $\mathrm{QEK}^{\mathrm{(glob)}}_{\boldsymbol{+}}$     & \textbf{\textcolor{BurntOrange}{89.45 $\pm$ 4.02}}     & 62.22 $\pm$ 5.54 \\
& $\mathrm{QEK}^{\mathrm{(glob)}}_{\boldsymbol{\times}}$& 88.96 $\pm$ 4.98     & 63.35 $\pm$ 4.78    \\
& $\mathrm{GDQC}^{\mathrm{(loc)}}_{\boldsymbol{+}}$     & \textbf{\textcolor{Cerulean}{89.35 $\pm$ 4.97}} & 65.86 $\pm$ 5.26    \\
& $\mathrm{GDQC}^{\mathrm{(loc)}}_{\boldsymbol{\times}}$& \textbf{\textcolor{JungleGreen}{90.83 $\pm$ 5.56}}     & \textbf{\textcolor{Cerulean}{66.51 $\pm$ 4.82}} \\
& $\mathrm{GDQC}^{\mathrm{(glob)}}_{\boldsymbol{+}}$    & 88.84 $\pm$ 5.42 & 59.82 $\pm$ 6.02 \\
& $\mathrm{GDQC}^{\mathrm{(glob)}}_{\boldsymbol{\times}}$ & 89.10 $\pm$ 5.55     & 61.18 $\pm$ 5.66    \\
\midrule
\end{tabular}

\label{tab:core_table}
\end{table}

\section{Results\label{sec:expe}}

In this section we present the extensive numerical results where we compare the various classical and quantum-feature kernel implementations introduced in Sec.~\ref{sec:kernels} with respect to their prediction performance using a support vector machine.
Table~\ref{tab:core_table} summarizes the results, providing a full comparison across all models and their respective best performances. 

We begin with a brief summary of the numerical-experiment setup. As classical baselines, we consider five widely used graph kernels (cf. Subsec.~\ref{subsec:classical_kernels}): the random-walk kernel, the graphlet (sampling) kernel~\cite{Przulj_2007}, the $k$-core kernel~\cite{k_core_kernels,Seidman_1983}, the shortest-path kernel, and the Weisfeiler-Lehman optimal assignment kernel~\cite{Kriege_2016}. 
The last three baselines can be used either with or without node features; when node features are available, we denote the corresponding variant by the superscript $\mathrm{(attr)}$.
Classical kernel baselines are computed using the GraKeL library~\cite{Grakel}.
The two quantum-feature kernels central to this study are the quantum evolution kernel (Subsec.~\ref{subsec:kernel_QEK}), which is based on the global excitation-probability observable $P_k^{\mathcal{G}}(t)$ (\ref{eq:observables}), and the generalized-distance quantum-correlation kernel (Subsec.~\ref{subsec:algo}), which is based on a local observable: the correlation matrix $C^{\mathcal{G}}_{ij}(t)$ (\ref{eq:observables}).

A key element of this work is the extension of the above quantum-feature kernels to attributed graphs. 
We achieve this by embedding atomic masses (node features) into the local detuning of the Hamiltonian $\hat{\mathcal{H}}_{\mathrm{loc}}^{\mathcal{G}}$ (\ref{eq:H_local}), in addition to embedding graph edges into atomic positions (cf. Subsec.~\ref{subsec:Graph_Embedding}). 
For comparison, we also consider a global-detuning Hamiltonian $\hat{\mathcal{H}}_{\mathrm{glob}}^{\mathcal{G}}$ (\ref{eq:H_global}), which embeds only edge information and uses no node features.
In Tabs.~\ref{tab:core_table}-\ref{tab:results} and Fig.~\ref{fig:pointwise}, we distinguish kernels based on whether the observables used to compute the kernel originate from time evolution under $\hat{\mathcal{H}}_{\mathrm{loc}}^{\mathcal{G}}$ or $\hat{\mathcal{H}}_{\mathrm{glob}}^{\mathcal{G}}$. 
We denote the corresponding variants by $\mathrm{QEK}^{(\mathrm{loc})}$ and $\mathrm{QEK}^{(\mathrm{glob})}$ (Eq.~\eqref{eq:QEK_kernel}), and analogously $\mathrm{GDQC}^{(\mathrm{loc})}$ and $\mathrm{GDQC}^{(\mathrm{glob})}$ (Eq.~\eqref{eq:GDQC_t}).
Finally, we study kernel pooling, where information from observables sampled at different times is combined (cf. Subsec.~\ref{subsec:pooling}). We consider sum pooling (Eq.~\eqref{eq:lin_comb}) and product pooling (Eq.~\eqref{eq:prod}). When pooling is applied, we add a subscript to the kernel name: $\boldsymbol{+}$ for sum pooling and $\boldsymbol{\times}$ for product pooling, e.g., $\mathrm{QEK}^{(\mathrm{loc})}_{\boldsymbol{+}}$ and $\mathrm{QEK}^{(\mathrm{loc})}_{\boldsymbol{\times}}$ (see Tabs.~\ref{tab:core_table}-\ref{tab:results}).

Throughout this work, we train an SVM classifier using precomputed kernel matrices (see the beginning of Sec.~\ref{sec:kernels} for details on kernel SVMs). We evaluate performance using the weighted F1 score, which accounts for class imbalance (notably in MUTAG and PTC\_FM$^{*}$). Concretely, the F1 score is computed per class and averaged with weights proportional to the class support.
We employ repeated stratified 5-fold cross-validation with 10 repetitions and report the mean and standard deviation of the weighted F1 score across all folds and repetitions. For each kernel, we tune the SVM regularization parameter $C$ over 100 logarithmically spaced values in $[10^{-3},10^{2}]$ and retain the best-performing model.
The best results are summarized in Tab.~\ref{tab:core_table}. Overall, incorporating node features improves performance. 
The best classical baseline is the WL optimal assignment kernel: with node features on PTC\_FM$^{*}$, and on unlabeled graphs for MUTAG.

\begin{figure}[t]
\centering
\includegraphics[width=1.\columnwidth]{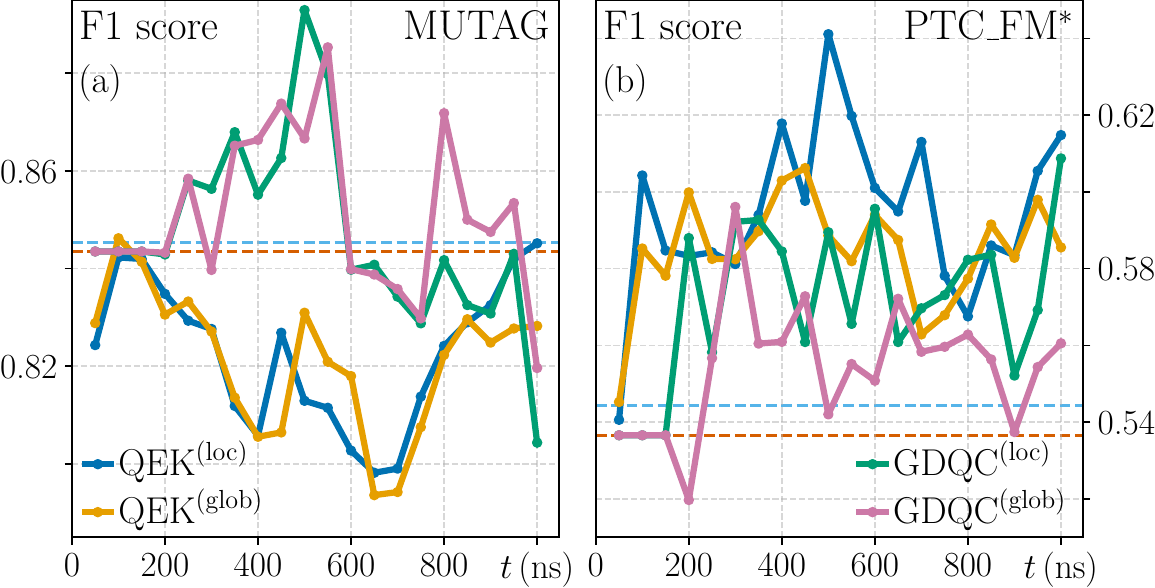}
\caption{Weighted F1 scores associated with SVMs trained on the quantum-feature kernels, $\mathrm{QEK}$ (\ref{eq:QEK_kernel}) and $\mathrm{GDQC}$ (\ref{eq:GDQC_t}), for the two datasets, (a) MUTAG and (b) PTC\_FM$^*$.
The superscripts $\mathrm{(loc)}$ and $\mathrm{(glob)}$ refer to whether the governing Hamiltonian has node features embedded [$\hat{\mathcal{H}}_{\mathrm{loc}}^{\mathcal{G}}$ (\ref{eq:H_local})] or not [$\hat{\mathcal{H}}_{\mathrm{glob}}^{\mathcal{G}}$ (\ref{eq:H_global})].
The training and evaluation procedure is separately conducted at each time steps across the time evolution. 
For comparison, the corresponding randomly stratified baseline, where each class is predicted proportionally to its frequency in the dataset are 0.56 for MUTAG and 0.53 for PTC\_FM$^*$, significantly below the presented F1 scores.
The $\mathrm{GDQC}$ results are derived for a number of bins $N_{\mathrm{bins}}^C = 10$.
The light blue and red horizontal dashed lines correspond to the classical early-time baselines for $\mathrm{QEK}$ [see Eq.~(\ref{eq:P_vector_early_time_approx})] and $\mathrm{GDQC}$ [see Eq.~(\ref{eq:chi_bin_classical})], respectively.
}
\label{fig:pointwise}
\end{figure}

Next, we focus on the quantum-feature kernels. Observables are sampled at $k=20$ equally spaced time steps $\{t_1,t_2,\dots,t_k=t_{\mathrm{max}}\}$. 
We consider two time-evolution durations: a shorter evolution with $t_{\mathrm{max}}=1\,\mu\mathrm{s}$ and a longer evolution with $t_{\mathrm{max}}=2\,\mu\mathrm{s}$. 
We begin with the non-pooled results. 
At each time step, a quantum-feature kernel matrix, $\mathrm{QEK}^{(\mathrm{loc}/\mathrm{glob})}$ or $\mathrm{GDQC}^{(\mathrm{loc}/\mathrm{glob})}$, is precomputed and used to train a SVM. 
The resulting weighted F1 scores for $t_{\mathrm{max}}=1\,\mu\mathrm{s}$ on both datasets (MUTAG and PTC\_FM$^{*}$) are shown in Fig.~\ref{fig:pointwise}. 
Overall, the quantum-feature kernels substantially outperform the random baseline and, for their best-performing variants, exceed the strongest classical baselines on both datasets (Tab.~\ref{tab:core_table}).

The horizontal dashed lines in Fig.~\ref{fig:dynamics} show the classical baselines obtained by evaluating the short-time perturbative expansions of the quantum features at a fixed time (Subsec.~\ref{subsec:perturbation_theory}): red for the $\mathrm{GDQC}$ baseline (cf.\ Eq.~\eqref{eq:chi_bin_classical}) and light blue for the $\mathrm{QEK}$ baseline (cf.\ Eq.~\eqref{eq:P_vector_early_time_approx}). 
At short times, $\mathrm{GDQC}$ behaves consistently with its early-time baseline. Only once non-negligible correlation elements have built up does the $\mathrm{GDQC}$ performance depart from this baseline. 
By contrast, particularly on MUTAG, the $\mathrm{QEK}$ kernel shows slightly weaker agreement with its early-time baseline. This discrepancy may stem from multiple factors, including finite-shot noise in the estimation of the observables and the coarse nature of the early-time approximation. 
In particular, already at the smallest sampled time step $t=40\,\mathrm{ns}$, the largest graphs in the dataset exhibit a non-negligible $P_2(t)$ contribution.

More generally, $\mathrm{QEK}$ and $\mathrm{GDQC}$ are not consistently comparable across all settings. In particular, $\mathrm{GDQC}$ attains the best score on MUTAG, whereas $\mathrm{QEK}$ achieves the best score on PTC\_FM$^{*}$ (Tab.~\ref{tab:core_table}). 
Regarding the use of the node-feature embedded Hamiltonian $\hat{\mathcal{H}}_{\mathrm{loc}}^{\mathcal{G}}$ versus the global-detuning Hamiltonian $\hat{\mathcal{H}}_{\mathrm{glob}}^{\mathcal{G}}$, we observe an overall trend in favor of local detuning across both datasets, with a particularly pronounced effect on PTC\_FM$^{*}$ (Tab.~\ref{tab:core_table}).
This pattern is consistent with the classical setting, where attributed variants often outperform their non-attributed counterparts. As the sampling time varies, performance changes only moderately, within $\sim 5$--$10\%$ in weighted F1 score, see Fig.~\ref{fig:pointwise}.
The best-performing time step (i.e., the maximum over the sampled times) is reported in Tab.~\ref{tab:core_table}$;$ note that it can originate from the time interval $1\,\mathrm{\mu s}< t \leq 2\,\mathrm{\mu s}$, which is not shown in Fig.~\ref{fig:pointwise}.

Next, we consider pooling schemes to further enhance the quantum-feature kernels. The two pooling schemes are described in Subsec.~\ref{subsec:pooling}. 
Since the choice of which of the $k=20$ time steps to pool (and how many to pool) is a priori arbitrary, we fix the tuple size and evaluate \emph{all} candidate time-step tuples considered in our sweep. 
The goal is to quantify the performance gains achievable through pooling and to assess the sensitivity of the models to the pooling choice. 
For each tuple size, we select the pooling configuration that maximizes cross-validated performance and report the corresponding test performance. 
Focusing on tuple sizes $3$ and $5$, Tab.~\ref{tab:results} summarizes the selected pooled models for $t_{\mathrm{max}}=1\,\mu\mathrm{s}$ and $t_{\mathrm{max}}=2\,\mu\mathrm{s}$, together with the chosen time steps.

Pooling can significantly improve performance: adding the best pooled models from each kernel class to Tab.~\ref{tab:core_table} yields an F1-score increase of $\sim 1\%$ on MUTAG and $\sim 5\%$ on PTC\_FM$^{*}$. 
Overall, the best-performing quantum-feature models outperform the classical baselines. 
Interestingly, the ranking of kernels is markedly dataset dependent: configurations that perform strongly on MUTAG do not necessarily remain top performers on PTC\_FM$^{*}$, and vice versa. 
Consistent with this observation, $\mathrm{GDQC}$ achieves the best performance on MUTAG, whereas $\mathrm{QEK}$ achieves the best performance on PTC\_FM$^{*}$ once pooling is applied (Tab.~\ref{tab:core_table}). 
A similar dataset dependence is also visible for classical baselines such as the random-walk and shortest-path kernels (Tabs.~\ref{tab:core_table}-\ref{tab:results}). 
For the quantum-feature kernels in particular, the choice of representation (local vs.\ global observables) and the pooling strategy can further shift the ordering (Tab.~\ref{tab:results}). We therefore treat these design choices as standard model-selection hyper-parameters to be selected via validation.

\begin{figure}[t]
\centering
\includegraphics[width=1.\columnwidth]{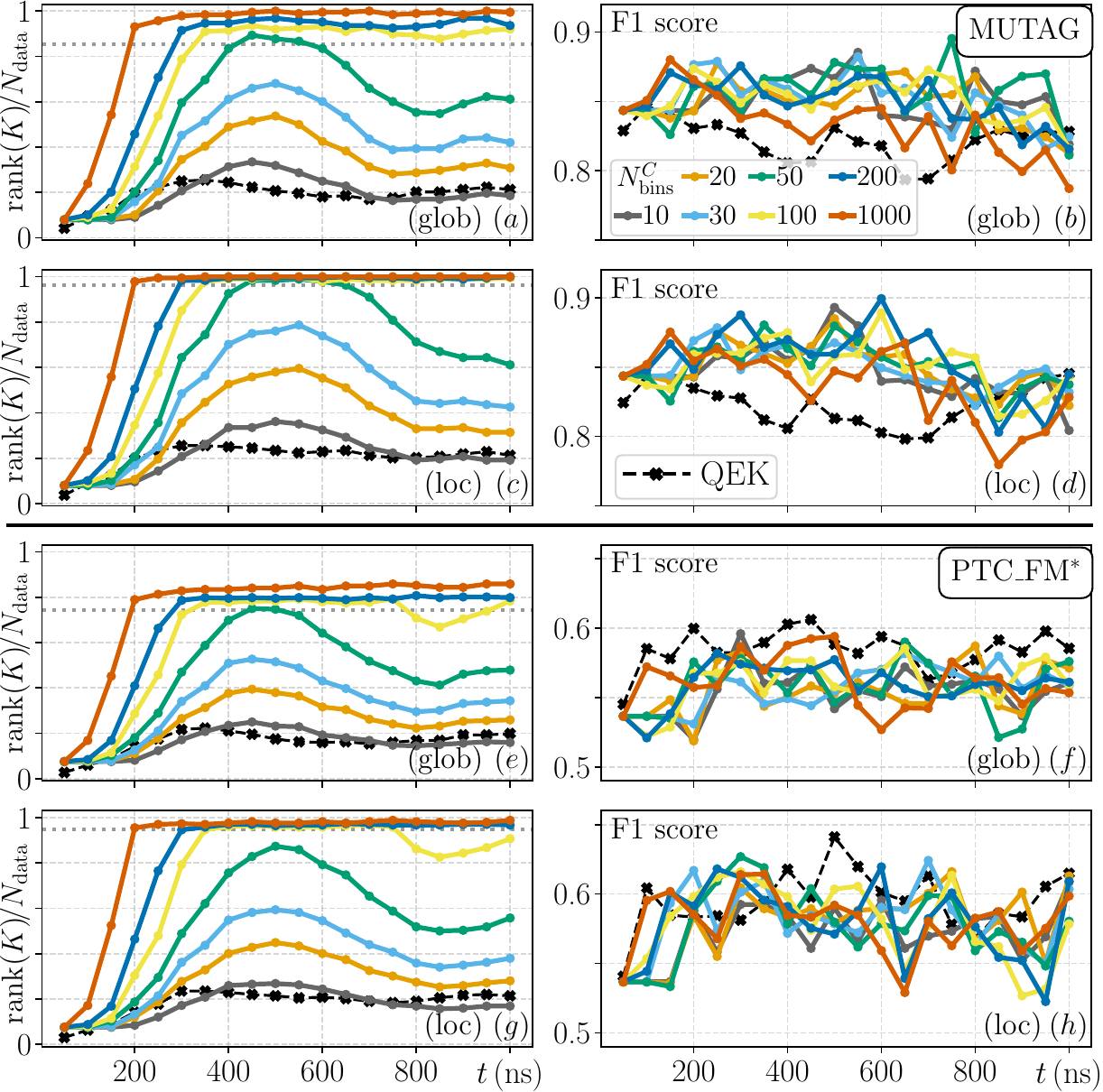} 
\caption{
Relative rank of the Gram matrix $\mathrm{rank}(K)/N_\mathrm{data}$ (\ref{eq:Gram_matrix}) associated with the $\mathrm{GDQC}$ kernel $\kappa_{\mathrm{GDQC}}$ (\ref{eq:GDQC_t}) versus the F1 score of the corresponding kernel-based SVM for varying binning number $N_{\mathrm{bins}}^{C}$.
The quantities are plotted as a function of the time steps across the entire time evolution with $t_\mathrm{max}=1\, \mathrm{\mu s}$.
(a)-(d) corresponds to the dataset MUTAG; (e)-(h) corresponds to PTC\_FM$^*$.
The identifier $\mathrm{(loc)}$ [$\mathrm{(glob)}$] imply that the involved time evolution is governed by the Hamiltonian $\hat{\mathcal{H}}_{\mathrm{loc}}^{\mathcal{G}}$ (\ref{eq:H_local}) [$\hat{\mathcal{H}}_{\mathrm{glob}}^{\mathcal{G}}$ (\ref{eq:H_global})], whereby $\hat{\mathcal{H}}_{\mathrm{loc}}^{\mathcal{G}}$ has node features embedded.
For comparison, the black dashed curve with x-markers represents the data associated with the $\mathrm{QEK}$ kernel $\kappa_{\mathrm{QEK}}$ (\ref{eq:QEK_kernel}), which does not depend on binning. The horizontal gray dotted line corresponds to the \textit{theoretical} relative rank value computed from the number of distinct graph-isomorphism classes. 
\label{fig:rank_vs_perf}}
\end{figure}

\begin{table*}[htbp]
\centering
\caption{
Weighted F1 scores associated with SVMs trained on different pooled quantum-feature kernel implementations.
The notation is identical to Tab.~\ref{tab:core_table}.
Shown are the best-performing models with pooling over $3$-tuples and $5$-tuples of time steps, for the two sets of sampled times with $t_\mathrm{max}=1\,\mathrm{\mu s}$ and $t_\mathrm{max}=2\,\mathrm{\mu s}$. 
The respective pooled time steps are included in the table.
The best-performing model for each kernel type is added to Tab.~\ref{tab:core_table}.
}
\begin{tabular}{c|c|cc|cc}
\toprule
\multirow{3}{*}{}  & \multirow{2}{*}{\textbf{Kernel}} & \multicolumn{2}{c|}{\textbf{MUTAG}} & \multicolumn{2}{c}{\textbf{PTC\_FM$^*$}} \\ 
\cmidrule{3-6}
&  & \textbf{time steps ($\mathrm{\mu s}$)} & \textbf{F1 Score (\%)} & \textbf{time steps ($\mathrm{\mu s}$)} & \textbf{F1 Score (\%)} \\
\midrule 

\multirow{18}{*}{\rotatebox{90}{maximum evolution time: $t_\mathrm{max}=1 \, \mathrm{\mu s}$}}  & \multirow{1}{*}{$\mathrm{QEK}^{\mathrm{(loc)}}$} & 1 & 84.52 $\pm$ 5.34 & 0.5 & 64.11 $\pm$ 4.75 \\

\cmidrule{2-6}
& \multirow{2}{*}{$\mathrm{QEK}^{\mathrm{(loc)}}_{\boldsymbol{\times}}$} & \{0.24, 0.44, 0.74\} & 87.05 $\pm$ 6.04 & \{0.14, 0.4, 1.\} & 66.54 $\pm$ 5.02 \\
&  & \{0.24, 0.3, 0.34, 0.74, 0.9\} & 87.61 $\pm$ 5.55 & \{0.1, 0.84, 0.9, 0.94, 1.\} & 66.66 $\pm$ 4.53 \\
\cmidrule{2-6}
&  \multirow{2}{*}{$\mathrm{QEK}^{\mathrm{(loc)}}_{\boldsymbol{+}}$} & \{0.24, 0.44, 0.64\} & 86.61 $\pm$ 5.97 & \{0.4, 0.94, 1\} & 67.34 $\pm$ 4.05 \\
&  & \{0.14, 0.24, 0.34, 0.44, 0.64\} & 86.89 $\pm$ 5.56 & \{0.04, 0.14, 0.34, 0.4, 1.\} & 67.29 $\pm$ 4.17 \\
\cmidrule{2-6}
& \multirow{1}{*}{$\mathrm{QEK}^{\mathrm{(glob)}}$} & 0.1 & 84.62 $\pm$ 5.48 & 0.44 & 60.62 $\pm$ 5.16 \\
\cmidrule{2-6}
& \multirow{2}{*}{$\mathrm{QEK}^{\mathrm{(glob)}}_{\boldsymbol{\times}}$} & \{0.14, 0.64, 0.7\} & 87.04 $\pm$ 4.46 & \{0.2, 0.24, 1.\} & 61.27 $\pm$ 5.05 \\
& & \{0.1, 0.24, 0.64, 0.7, 0.74\} & 87.35 $\pm$ 5.15 & \{0.14, 0.2, 0.34, 0.4, 1.\} & 62.6 $\pm$ 5.79 \\
\cmidrule{2-6}
& \multirow{2}{*}{$\mathrm{QEK}^{\mathrm{(glob)}}_{\boldsymbol{+}}$} & \{0.2, 0.64, 0.7\} & 87.47 $\pm$ 4.9 & \{0.24, 0.34, 1\} & 60.78 $\pm$ 4.72 \\
& & \{0.1, 0.24, 0.64, 0.7, 0.74\} & 88.21 $\pm$ 5.92 & \{0.2, 0.24, 0.34, 0.4, 1.\} & 61.19 $\pm$ 4.89 \\
\cmidrule{2-6}

& \multirow{1}{*}{$\mathrm{GDQC}^{\mathrm{(loc)}}$} & 0.5 & 89.29 $\pm$ 4.89 & 1 & 60.87 $\pm$ 5.49 \\
\cmidrule{2-6}
&  \multirow{2}{*}{$\mathrm{GDQC}^{\mathrm{(loc)}}_{\boldsymbol{\times}}$} & \{0.5, 0.74, 0.94\} & 89.58 $\pm$ 4.81 & \{0.4, 0.84, 1.\} & 66.41 $\pm$ 5.70 \\
&  & \{0.24, 0.5, 0.64, 0.7, 0.94\} & \textbf{\textcolor{Cerulean}{90.09 $\pm$ 5.05}} & \{0.24, 0.34, 0.4, 0.84, 1.\} & 66.90 $\pm$ 4.58 \\
\cmidrule{2-6}
& \multirow{2}{*}{$\mathrm{GDQC}^{\mathrm{(loc)}}_{\boldsymbol{+}}$} & \{0.04, 0.5, 0.64\} & 88.74 $\pm$ 5.69 & \{0.34, 0.84, 1.\} & 64.09 $\pm$ 4.97 \\
& & \{0.3, 0.54, 0.6, 0.74, 0.9\} & 88.81 $\pm$ 5.55 & \{0.3, 0.5, 0.64, 0.84, 1.\} & 65.69 $\pm$ 4.86 \\
\cmidrule{2-6}

& \multirow{1}{*}{$\mathrm{GDQC}^{\mathrm{(glob)}}$} & 0.54 & 88.53 $\pm$ 4.33 & 0.3 & 59.61 $\pm$ 5.57 \\
\cmidrule{2-6}
& \multirow{2}{*}{$\mathrm{GDQC}^{\mathrm{(glob)}}_{\boldsymbol{\times}}$} & \{0.24, 0.4, 0.74\} & 88.95 $\pm$ 5.50 & \{0.04, 0.1, 0.24\} & 61.19 $\pm$ 5.44 \\  
& & \{0.04, 0.1, 0.24, 0.4, 0.54\} & 88.82 $\pm$ 5.07 & \{0.2, 0.6, 0.64, 0.8, 0.84\} & 61.25 $\pm$ 5.61 \\
\cmidrule{2-6}
& \multirow{2}{*}{$\mathrm{GDQC}^{\mathrm{(glob)}}_{\boldsymbol{+}}$} & \{0.24, 0.4, 0.54\} & 88.35 $\pm$ 5.18 & \{0.24, 0.3, 0.84\} & 59.32 $\pm$ 4.84 \\
& & \{0.04, 0.2, 0.24, 0.4, 0.54\} & 88.60 $\pm$ 5.27 & \{0.24, 0.4, 0.6, 0.64, 0.8\} & 59.43 $\pm$ 4.39 \\
\midrule

\multirow{18}{*}{\rotatebox{90}{maximum evolution time: $t_\mathrm{max}=2 \, \mathrm{\mu s}$}} & \multirow{1}{*}{$\mathrm{QEK}^{\mathrm{(loc)}}$} & 2 & 86.19 $\pm$ 5.03 & 0.5 & 64.11 $\pm$ 4.75 \\
\cmidrule{2-6}
& \multirow{2}{*}{$\mathrm{QEK}^{\mathrm{(loc)}}_{\boldsymbol{\times}}$} & \{1.3, 1.4, 1.8\} & 88.92 $\pm$ 4.59 & \{0.1, 1.1, 1.6\} & 68.32 $\pm$ 5.44 \\
& & \{0.1, 1.1, 1.3, 1.4, 2\} & 88.78 $\pm$ 4.89  & \{0.4, 1.1, 1.3, 1.5, 1.6\} & \textbf{\textcolor{BurntOrange}{69.28 $\pm$ 5.28}} \\
\cmidrule{2-6}
& \multirow{2}{*}{$\mathrm{QEK}^{\mathrm{(loc)}}_{\boldsymbol{+}}$} & \{1.4, 1.6, 1.7\} & 88.16 $\pm$ 5.99 & \{0.2, 1.1, 1.5\} & \textbf{\textcolor{Cerulean}{68.65 $\pm$ 5.96}} \\
& & \{0.1, 0.7, 1.4, 1.5, 1.8\} & 89.26 $\pm$ 5.40 & \{1., 1.1, 1.3, 1.5, 1.9\} & \textbf{\textcolor{JungleGreen}{69.55 $\pm$ 4.98}} \\
\cmidrule{2-6}
& \multirow{1}{*}{$\mathrm{QEK}^{\mathrm{(glob)}}$} & 2 & 85.8 $\pm$ 5.81 & 0.44 & 60.62 $\pm$ 5.16 \\
\cmidrule{2-6}
&  \multirow{2}{*}{$\mathrm{QEK}^{\mathrm{(glob)}}_{\boldsymbol{\times}}$} & \{0.3, 0.5, 1.5\} & 88.57 $\pm$ 4.03 & \{1., 1.3, 1.5\} & 62.56 $\pm$ 5.87 \\
& & \{0.2, 0.5, 1.1, 1.4, 1.5\} & 88.96 $\pm$ 4.98 & \{0.3, 1.1, 1.2, 1.3, 1.5\} & 63.35 $\pm$ 4.78 \\
\cmidrule{2-6}
& \multirow{2}{*}{$\mathrm{QEK}^{\mathrm{(glob)}}_{\boldsymbol{+}}$}& \{0.2, 1.5, 2.\} & 87.63 $\pm$ 5.58 & \{0.2, 0.4, 1.9\} & 62.66 $\pm$ 4.92 \\
& & \{0.2, 0.5, 0.7, 1.4, 1.7\} & 89.45 $\pm$ 4.02 & \{0.3, 0.4, 1., 1.8, 1.9\} & 62.22 $\pm$ 5.54 \\
\cmidrule{2-6}

& \multirow{1}{*}{$\mathrm{GDQC}^{\mathrm{(loc)}}$} & 0.5 & 89.29 $\pm$ 4.89 & 1 & 60.87 $\pm$ 5.49 \\
\cmidrule{2-6}
& \multirow{2}{*}{$\mathrm{GDQC}^{\mathrm{(loc)}}_{\boldsymbol{\times}}$} & \{1.3, 1.4, 1.7\} & \textbf{\textcolor{BurntOrange}{90.25 $\pm$ 5.02}} & \{0.04, 0.84, 1.\} & 66.41 $\pm$ 5.70 \\
& & \{0.1, 0.5, 1.3, 1.4, 1.7\} & \textbf{\textcolor{JungleGreen}{90.83 $\pm$ 5.56}} &
\{0.4, 1.1, 1.5, 1.6, 1.7\} & 66.51 $\pm$ 4.82 \\
\cmidrule{2-6}
& \multirow{2}{*}{$\mathrm{GDQC}^{\mathrm{(loc)}}_{\boldsymbol{+}}$} & \{0.5, 1.3, 1.9\} & 89.28 $\pm$ 4.21 & \{0.3, 1., 1.9\} & 65.48 $\pm$ 4.21 \\
& & \{0.5, 0.6, 1.1, 1.3, 1.8\} & 89.35 $\pm$ 4.97 & \{0.1, 0.3, 1., 1.3, 1.6\} & 65.86 $\pm$ 5.26 \\
\cmidrule{2-6}

& \multirow{1}{*}{$\mathrm{GDQC}^{\mathrm{(glob)}}$} & 0.8 & 88.53 $\pm$ 4.33 & 0.3 & 59.61 $\pm$ 5.57 \\
\cmidrule{2-6}
& \multirow{2}{*}{$\mathrm{GDQC}^{\mathrm{(glob)}}_{\boldsymbol{\times}}$} & \{0.8, 1.2, 1.3\} & 88.65 $\pm$ 5.56 & \{0.6, 1.2, 1.7\} & 60.4 $\pm$ 4.85 \\
& & \{0.5, 1.3, 1.4, 1.7, 1.9\} & 89.10 $\pm$ 5.55 & \{0.3, 0.6, 1.1, 1.4, 1.7\} & 61.18 $\pm$ 5.66 \\
\cmidrule{2-6}
& \multirow{2}{*}{$\mathrm{GDQC}^{\mathrm{(glob)}}_{\boldsymbol{+}}$} & \{0.8, 1.2, 1.3\} & 88.84 $\pm$ 5.42 & \{0.1, 1.4, 1.5\} & 59.41 $\pm$ 5.23 \\
& & \{0.8, 0.9, 1.2, 1.3, 1.9\} & 88.84 $\pm$ 5.48 & \{0.3, 0.6, 1., 1.1, 1.7\} & 59.83 $\pm$ 6.02 \\
\midrule
\end{tabular}
\label{tab:results}
\end{table*}

\subsection{Binning-size dependence, graph distinguishability and limitations \label{subsec:binning_analysis}}

In this part, we study the impact of the correlation-binning resolution $\varDelta_C = 1/N_{\mathrm{bins}}^C$ applied to the correlation matrix $C^{\mathcal{G}}_{i,j}(t)$ within the $\mathrm{GDQC}$ framework (cf. Subsec.~\ref{subsec:algo}).

To this end, we vary the number of correlation bins $N_{\mathrm{bins}}^{C}$ and track two complementary quantities. 
First, we report the weighted F1 score, which captures the predictive performance of the corresponding $\mathrm{GDQC}$ model. 
Second, we compute the rank of the associated Gram matrix $K \in \mathbb{R}^{N_{\mathrm{data}}\times N_{\mathrm{data}}}$ with entries
\begin{align}
K_{\mu,\mu'} &= \kappa_{\mathrm{GDQC}}\!\left(\mathcal{G}_{\mu},\mathcal{G}_{\mu'}\right),
& \mu,\mu' &\in \{1,\dots,N_{\mathrm{data}}\},
\label{eq:Gram_matrix}
\end{align}
where $N_{\mathrm{data}}$ denotes the dataset size, and the $\mathrm{GDQC}$ kernel $\kappa_{\mathrm{GDQC}}$ is defined in Eq.~(\ref{eq:GDQC_t}). 
We estimate $\mathrm{rank}(K)$ by counting the number of singular values above a numerical threshold~\cite{NumRec,numpy_matrix_rank_v2_4}. The normalized quantity $\mathrm{rank}(K)/N_{\mathrm{data}} \in [0,1]$ serves as a proxy for the effective dimensionality of the kernel representation on the sampled dataset, and hence for the degree to which the kernel induces non-degenerate (linearly independent) feature representations.
Figure~\ref{fig:rank_vs_perf} reports both the relative rank $\mathrm{rank}\left( K \right)/N_\mathrm{data}$ and the weighted F1 score for varying $N_{\mathrm{bins}}^{C}$, across the time evolution up to $t_\mathrm{max}=1\,\mathrm{\mu s}$.

We observe that increasing the number of correlation bins $N_{\mathrm{bins}}^{C}$ increases the rank of the $\mathrm{GDQC}$ Gram matrix, i.e., the effective dimensionality of $\kappa_{\mathrm{GDQC}}$ in Eq.~\eqref{eq:GDQC_t}. This suggests that, on the present datasets, finer binning yields a less degenerate similarity structure and allows the kernel to separate a larger number of nearby graphs. Beyond a certain $N_{\mathrm{bins}}^{C}$, the relative rank $\mathrm{rank}\left( K \right)/N_\mathrm{data}$ saturates. 
The saturation level is close to $1$ for the local-detuning variant $\mathrm{GDQC}^{(\mathrm{loc})}$ (panels (c),(g)), whereas it remains substantially smaller for the global-detuning variant $\mathrm{GDQC}^{(\mathrm{glob})}$ (panels (a),(e)), consistent with local detuning inducing a richer empirical similarity structure than global detuning on these datasets.

To interpret these rank values, we also compute an idealized reference rank from the number and cardinalities of graph isomorphism classes (shown as the gray dotted line in Fig.~\ref{fig:rank_vs_perf}), where isomorphism classes are determined using the VF2 algorithm~\cite{cordella_2004}. 
This reference depends on whether node attributes are included in the isomorphism test: ignoring attributes yields fewer distinct classes and hence a smaller reference rank than the attribute-aware case, which is expected since some molecules differ only by atom and bond types.
Notably, for a large number of correlation bins (fine binning), the empirical rank of the $\mathrm{GDQC}$ Gram matrix can exceed this combinatorial reference value. This reflects a limitation of our current pipeline. Indeed, the embedding computation involves stochasticity through the random initialization of atomic positions. As a result, repeated embeddings of the same graph do not necessarily converge to identical positions, which can lead to slightly different quantum features. This effect is present for both $\mathrm{GDQC}$ and $\mathrm{QEK}$, but $\mathrm{GDQC}$ is more sensitive due to the binning step: under fine binning, small variations in correlation values can be amplified.
For moderate binning, i.e., $N_{\mathrm{bins}}^{C} \lesssim 100$, the induced RKHS representations of two isomorphic graphs remain very close. For normalized kernels, this can be quantified via
\begin{align}
\left\|\boldsymbol{\chi}(\mathcal{G}_{\mu})-\boldsymbol{\chi}(\mathcal{G}_{\mu^{\prime}})\right\|^{2}
&= 2 - 2K_{\mu, \mu^{\prime}} \ll 1.
\label{eq:feature_vector_proximity}
\end{align}
This suggests that the perturbation induced by stochastic embeddings is limited in that regime. Quantifying its impact on downstream SVM performance would require an explicit robustness analysis across random initializations, which we leave for future work.

For completeness, we perform the same Gram-rank analysis for the $\mathrm{QEK}$ kernel $\kappa_{\mathrm{QEK}}$ (\ref{eq:QEK_kernel}) and include the corresponding results in Fig.~\ref{fig:rank_vs_perf} (black dashed line with $\times$ markers). 
Recall that $\mathrm{QEK}$ has no binning dependence. 
Interestingly, the rank of $\kappa_{\mathrm{QEK}}$ is comparable to that of $\kappa_{\mathrm{GDQC}}$ at only $N_{\mathrm{bins}}^{C}=10$. 
In the present experimental setting, $\mathrm{GDQC}$ therefore offers a broader tunable range of empirical graph distinguishability through the binning parameter. However, this increased distinguishability does not yield a clear improvement in weighted F1 score in our experiments (cf. Fig.~\ref{fig:pointwise}).

This finding highlights that \emph{graph distinguishability} is a representational property of the induced feature space (as reflected by Gram-matrix characteristics such as rank, cf. Fig.~\ref{fig:rank_vs_perf}), but it is not sufficient to predict downstream accuracy.
Predictive performance also depends on how well kernel-induced similarities align with the labels under the chosen learning procedure, including model selection, regularization, and finite-sample effects, as well as (for quantum features) discretization and finite-shot noise.
Hence, increased distinguishability may help, but does not by itself guarantee improved generalization. 
In particular, two detrimental extremes can arise: if a kernel is overly discriminative on a dataset with no isomorphic pairs, the Gram matrix approaches the identity and provides little usable similarity structure. Conversely, if a kernel assigns nearly identical similarity to all graph pairs, the Gram matrix approaches the all-ones matrix, which is likewise uninformative.

\section{Conclusion and Perspectives\label{sec:conclusion}}

In this work, we investigate different quantum-feature kernel implementations~\eqref{eq:general_quantum_evolution_kernel} and evaluate their performance on two benchmark datasets using support vector machines. 
First, we introduce a graph-embedding scheme that encodes edge information into a neutral-atom Rydberg Hamiltonian~\eqref{eq:Rydberg_Ham}, and additionally incorporates node features (atomic masses) via an appropriate choice of local detuning fields (cf. Subsec.~\ref{subsec:Graph_Embedding}). 
This allows us to compare node-feature embedding (time evolution under $\hat{\mathcal{H}}_{\mathrm{loc}}^{\mathcal{G}}$, Eq.~\eqref{eq:H_local}) with a non-attributed construction (time evolution under $\hat{\mathcal{H}}_{\mathrm{glob}}^{\mathcal{G}}$, Eq.~\eqref{eq:H_global}). 
While node-feature embedding increases kernel expressiveness both theoretically (Prop.~\ref{prop:prop1}) and empirically (Fig.~\ref{fig:rank_vs_perf}), we observe comparable SVM prediction performance in the two settings (Fig.~\ref{fig:pointwise}), with a slight advantage for the local addressability approach.

Second, we compare two conceptually different quantum-feature kernels: 
the quantum evolution kernel $\kappa_{\mathrm{QEK}}$ (Eq.~\eqref{eq:QEK_kernel})~\cite{Henry_2021,Albrecht2023}, which is based on the global excitation-probability observable $P_k^{\mathcal{G}}(t)$ (Eq.~\eqref{eq:observables}), 
and the proposed generalized-distance quantum-correlation kernel $\kappa_{\mathrm{GDQC}}$ (Eq.~\eqref{eq:GDQC_t}), which is based on the local correlation matrix $C_{i,j}^{\mathcal{G}}(t)$ (Eq.~\eqref{eq:observables}). 
While these kernels can, in principle, behave differently depending on the target task, for the global classification tasks considered here they achieve comparable performance and are competitive with strong classical baselines (Tab.~\ref{tab:core_table}).

Third, we improve performance by pooling information from observables sampled at multiple times across the evolution (cf. Subsec.~\ref{subsec:pooling}). 
Pooling can significantly enhance model performance, enabling the quantum-feature kernels to surpass the classical baselines (Tabs.~\ref{tab:core_table}-\ref{tab:results}). 
Overall, the proposed implementations broaden the scope of existing approaches in quantum graph machine learning.

Our quantum kernels do not provide an asymptotic speedup over the classical baselines considered here: $\mathrm{QEK}$ scales as $\mathcal{O}(N_{\mathrm{shots}}\,N)$, while $\mathrm{GDQC}$ requires $\mathcal{O}(N_{\mathrm{shots}}\,N^2)$. 
By comparison, the best-performing classical baseline in our experiments is the WL optimal assignment kernel, with complexity $\mathcal{O}(h\,M + N)$. 
Moreover, the cost of analog quantum evolution is difficult to quantify in standard algorithmic terms, since complexity notions are inherently tied to digital computation. 
In this work, the value proposition therefore lies in empirical performance (competitive weighted F1 scores on both benchmarks) rather than an asymptotic runtime advantage.

Finally, hyper-parameter sensitivity remains an important open question. 
In particular, we find it difficult to draw definitive conclusions about global performance trends as a function of parameters such as the local detuning scale $\delta_0$ (Eq.~\eqref{eq:local_detuning}), the Rabi frequency $\Omega$ (Eq.~\eqref{eq:Rydberg_Ham}), the $\mathrm{GDQC}$ binning resolution (Alg.~\ref{algo:GDQC}), and the choice of sampling times (with or without pooling).

\begin{acknowledgments}
The authors gratefully acknowledge Louis-Paul Henry, Vittorio Vitale, and Francesco Musso for the many fruitful discussions that have significantly contributed to the development of this work.
We also thank Anton Quelle, Kemal Bidzhiev, Mauro Mendizabal, and Elie Merhej for their valuable contributions to the development of the emulation tools, which were essential for the implementation of the algorithms presented in this work.
\end{acknowledgments}

\appendix

\section{Time-dependent perturbation theory\label{app:perturbation theory}}

In this appendix, we revisit basics on the time-dependent perturbation theory, and we derive the leading observable contributions which are presented in Subsec.~\ref{subsec:perturbation_theory}.

We begin with a generic system Hamiltonian 
\begin{align}
\hat{\mathcal{H}}\left(t\right) & =\hat{H}_{0}+\hat{V}\left(t\right),\label{eq:system_Hamiltonian}
\end{align}
decomposed into an unperturbed part $\hat{H}_{0}$ (assumed to be
time independent, for simplicity), and a time-dependent perturbation
part $\hat{V}\left(t\right)$ which is considered to be small. 
In this part we distinguish between the Schr\"odinger and the interaction
picture, which we label via subscripts $\mathrm{S}$ and $\mathrm{I}$,
respectively. 
In the Schr\"odinger picture, the time-evolved state is given by 
\begin{align}
\left|\psi_{\mathrm{S}}\left(t\right)\right\rangle  & =\hat{U}_{\mathrm{S}}\left(t,t_{0}\right)\left|\psi_{\mathrm{S}}\left(t_{0}\right)\right\rangle ,\label{eq:Schroedinger_evolution}
\end{align}
with the Schr\"odinger time evolution operator 
\begin{align}
\hat{U}_{\mathrm{S}}\left(t,t_{0}\right) & =\hat{\mathcal{T}}_{t}\,\exp\left[-\frac{\mathsf{i}}{\hbar}\int_{t_{0}}^{t}dt^{\prime}\;\hat{\mathcal{H}}\left(t^{\prime}\right)\right]\label{eq:schroedinger_evolution_operator}\\
 & =e^{-\frac{\mathsf{i}}{\hbar}\hat{H}_{0}t}\hat{U}_{\mathrm{I}}\left(t,t_{0}\right)e^{\frac{\mathsf{i}}{\hbar}\hat{H}_{0}t_{0}},\label{eq:U_S_U_i_relation}
\end{align}
whose relation to the interaction-picture time evolution operator
$\hat{U}_{\mathrm{I}}\left(t,t_{0}\right)$ is shown in Eq.~(\ref{eq:U_S_U_i_relation}).
The Schr\"odinger equation in the interaction picture becomes 
\begin{align}
\mathsf{i}\hbar\frac{\partial}{\partial t}\hat{U}_{\mathrm{I}}\left(t,t_{0}\right) & =\hat{V}_{\mathrm{I}}\left(t\right)\hat{U}_{\mathrm{I}}\left(t,t_{0}\right),\label{eq:U_I_SEQ}
\end{align}
with 
\begin{align}
\hat{V}_{\mathrm{I}}\left(t\right) & =e^{\frac{\mathsf{i}}{\hbar}\hat{H}_{0}t}\hat{V}\left(t\right)e^{-\frac{\mathsf{i}}{\hbar}\hat{H}_{0}t},\label{eq:V_I}
\end{align}
and the formal solution 
\begin{align}
\hat{U}_{\mathrm{I}}\left(t,t_{0}\right) & =\hat{\mathcal{T}}_{t}\;\exp\left[-\frac{\mathsf{i}}{\hbar}\int_{t_{0}}^{t}dt^{\prime}\,\hat{V}_{\mathrm{I}}\left(t^{\prime}\right)\right],\label{eq:evolution_op_dyson_eq}
\end{align}
where $\hat{\mathcal{T}}_{t}$ is the time ordering operator. 
At this point we assume that the exponent in (\ref{eq:evolution_op_dyson_eq}) is effectively small
such that the expansion
\begin{align}
\hat{U}_{\mathrm{I}}\left(t,t_{0}\right) & =\sum_{l=0}^{\infty}\hat{U}_{\mathrm{I}}^{\left(l\right)}\left(t,t_{0}\right),\label{eq:Neumann_series}
\end{align}
in terms of a Neumann series is justified.
The $l$-th series coefficient is given by
\begin{align}
\hat{U}_{\mathrm{I}}^{\left(l\right)}\left(t,t_{0}\right) & =\left(-\frac{\mathsf{i}}{\hbar}\right)^{l}\int_{t_{0}}^{t}dt_{1}\int_{t_{0}}^{t_{1}}dt_{2}\cdots\int_{t_{0}}^{t_{l-1}}dt_{l}\,\times\nonumber \\
 & \qquad\hat{V}_{\mathrm{I}}\left(t_{1}\right)\hat{V}_{\mathrm{I}}\left(t_{2}\right)\cdots\hat{V}_{\mathrm{I}}\left(t_{l}\right),\label{eq:series_coeffs_Un}
\end{align}
where $\hat{U}_{\mathrm{I}}^{\left(0\right)}\left(t,t_{0}\right)=\mathbbm{1}$.
For practical purposes, it is often convenient to exploit the inductive relation
\begin{align}
\hat{U}_{\mathrm{I}}^{\left(l\right)}\left(t,t_{0}\right) & =-\frac{\mathsf{i}}{\hbar}\int_{t_{0}}^{t}dt^{\prime}\;\hat{V}_{\mathrm{I}}\left(t^{\prime}\right)\;\hat{U}_{\mathrm{I}}^{\left(l-1\right)}\left(t^{\prime},t_{0}\right),\label{eq:U_k_induction}
\end{align}
for $l\geq1$. 
Switching back to the Schr\"odinger picture via relation (\ref{eq:U_S_U_i_relation})
we can finally expand the time-evolved wave function (\ref{eq:Schroedinger_evolution})
according to
\begin{align}
\left|\psi_{\mathrm{S}}\left(t\right)\right\rangle  & =\sum_{l=0}^{\infty}\left|\psi_{\mathrm{S}}^{\left(l\right)}\left(t\right)\right\rangle ,\label{eq:wave_function_series}
\end{align}
where the $l$-th order wave-function contribution is given by
\begin{align}
\left|\psi_{\mathrm{S}}^{\left(l\right)}\left(t\right)\right\rangle  & =e^{-\frac{\mathsf{i}}{\hbar}\hat{H}_{0}t}\hat{U}_{\mathrm{I}}^{\left(l\right)}\left(t,t_{0}\right)e^{\frac{\mathsf{i}}{\hbar}\hat{H}_{0}t_{0}}\left|\psi_{\mathrm{S}}\left(t_{0}\right)\right\rangle .\label{eq:wave_function_coeffs}
\end{align}
In practice, the computation can be facilitated using the induction
rule (\ref{eq:U_k_induction}) such that 
\begin{align}
\left|\psi_{\mathrm{S}}^{\left(l\right)}\left(t\right)\right\rangle  & =-\frac{\mathsf{i}}{\hbar}\int_{t_{0}}^{t}dt^{\prime}\;e^{-\frac{\mathsf{i}}{\hbar}\hat{H}_{0}\left(t-t^{\prime}\right)}\hat{V}\left(t^{\prime}\right)\,\left|\psi_{\mathrm{S}}^{\left(l-1\right)}\left(t^{\prime}\right)\right\rangle ,\label{eq:psi_S_induction}
\end{align}
for $l\geq1$ and $|\psi_{\mathrm{S}}^{\left(0\right)}\left(t\right)\rangle =\exp\left[-\frac{\mathsf{i}}{\hbar}\hat{H}_{0}\left(t-t_{0}\right)\right]\left|\psi_{\mathrm{S}}\left(t_{0}\right)\right\rangle $.
Now, wave-function contributions and eventually, observables can be
evaluated order by order. Typically, the evaluation of (\ref{eq:psi_S_induction})
is conducted using the eigenbasis of the unperturbed Hamiltonian $\hat{H}_{0}$.

Starting from the Rydberg Hamiltonian (\ref{eq:Rydberg_Ham}) we consider the transverse
part as the perturbation, i.e. 
\begin{align}
\hat{H}_{0} & =\frac{1}{2} \sum\nolimits_{i,j}^{\prime} J_{ij}\hat{n}_{i}\hat{n}_{j}-\hbar\sum_{i=1}^{N}\delta_{i}\hat{n}_{i}, & \hat{V} & =\frac{\hbar\Omega}{2}\sum_{i=1}^{N}\hat{\sigma}_{i}^{x}.\label{eq:H0-1}
\end{align}
Note that $\hat{V}$ is time independent in this case. 
The goal is to derive the leading contributions in powers of $\left(\Omega t\right)$ for the
observables of interest (\ref{eq:observables}), the excitation probability $P_{k}\left(t\right)$,
the site occupation $n_{i}\left(t\right)$, and the correlation matrix
$C_{i,j}\left(t\right)$. 
Recall that our initial state at time $t_0=0$ is the all-zero
state, i.e. $\left|\psi\left(t_{0}=0\right)\right\rangle =\left|\boldsymbol{0}\right\rangle $.
By construction, the $l$-th order wave-function
term $|\psi_{\mathrm{S}}^{\left(l\right)}\left(t\right)\rangle $
scales as $\mathcal{O}\left((\Omega t)^{l}\right)$, see Eq.~(\ref{eq:psi_S_induction}). 
Furthermore, because the perturbation $\hat{V}$ (\ref{eq:H0-1}) can maximally cause one excitation per application, the contribution $|\psi_{\mathrm{S}}^{\left(l\right)}\left(t\right)\rangle $
can maximally contain $l$ excitations.
For later convenience, we explicitly evaluate the first three wave-function contributions with $l \in \left\{0,1,2 \right\}$. 
Using the relation (\ref{eq:psi_S_induction}), one finds $|\psi_{\mathrm{S}}^{\left(0\right)}\left(t\right)\rangle =\left|\boldsymbol{0}\right\rangle $,
and
\begin{align}
\left|\psi_{\mathrm{S}}^{\left(1\right)}\left(t\right)\right\rangle  & =\frac{\Omega}{2}\sum_{i=1}^{N}\frac{1-e^{\mathsf{i}\delta_{i}t}}{\delta_{i}}\,\hat{\sigma}_{i}^{x}\left|\boldsymbol{0}\right\rangle ,\label{eq:psi_S_1}\\
\left|\psi_{\mathrm{S}}^{\left(2\right)}\left(t\right)\right\rangle  & =a_{0}^{\left(2\right)}\,\left|\boldsymbol{0}\right\rangle + \sum\nolimits_{i,j}^{\prime} a_{ij}^{\left(2\right)}\,\hat{\sigma}_{j}^{x}\,\hat{\sigma}_{i}^{x}\,\left|\boldsymbol{0}\right\rangle ,\label{eq:psi_S_2}
\end{align}
where we defined the second-order amplitudes 
\begin{align}
a_{0}^{\left(2\right)} & =\frac{\Omega^{2}}{4}\sum_{i=1}^{N}\,\frac{e^{\mathsf{i}\delta_{i}t}-1-\mathsf{i}\delta_{i}t}{\delta_{i}^{2}},\label{eq:a_(2)_0}\\
a_{ij}^{\left(2\right)} & =-\frac{\Omega^{2}}{4\delta_{i}}\left[\frac{1-e^{-\mathsf{i}t\omega_{ij}^{\left(2\right)}}}{\omega_{ij}^{\left(2\right)}}-\frac{e^{\mathsf{i}\delta_{i}t}-e^{-\mathsf{i}\omega_{ij}^{\left(2\right)}t}}{\omega_{ij}^{\left(2\right)}+\delta_{i}}\right],\label{eq:a_(2)_ij}
\end{align}
with $\omega_{ij}^{\left(2\right)}=(J_{ij}/\hbar) -\delta_{i}-\delta_{j}$.

Next, we evaluate the observables, starting with the excitation
probability operator $\hat{P}_{k}=\sum_{\boldsymbol{b}}\delta_{k,\sum_{i}b_{i}}\left|\boldsymbol{b}\right\rangle \left\langle \boldsymbol{b}\right|$,
cf. Eq.~(\ref{eq:observables}). 
Since the index $k$ effectively counts the number of excitations, it is clear that the first non-vanishing projection occurs for the wave-function contribution $|\psi_{\mathrm{S}}^{\left(l=k\right)}\left(t\right)\rangle $, leading to
\begin{align}
P_{k}\left(t\right) & =\sum_{l,l^{\prime}=k}^{\infty}\Big\langle \psi_{\mathrm{S}}^{(l^{\prime})}\left(t\right)\Big|\hat{P}_{k}\Big|\psi_{\mathrm{S}}^{\left(l\right)}\left(t\right)\Big\rangle =\mathcal{O}\left(\left(\Omega t\right)^{2k}\right).\label{eq:P_k_orders}
\end{align}
In this work, we only include contributions up to order $\mathcal{O}\left(\left(\Omega t\right)^{2}\right)$.
Hence, the two non-negligible terms are $P_{0}\left(t\right)$ and
$P_{1}\left(t\right)$ which can be unpacked as 
\begin{align}
P_{0}\left(t\right) & =1+2\Re\left[\left\langle \boldsymbol{0}\big|\psi_{\mathrm{S}}^{\left(2\right)}\left(t\right)\right\rangle \right]+\mathcal{O}\left(\left(\Omega t\right)^{4}\right),\label{eq:P_0_app_fin}\\
P_{1}\left(t\right) & =\sum_{i=1}^{N}\left|\left\langle \boldsymbol{0}\,\big|\,\hat{\sigma}_{i}^{x}\,\big|\,\psi_{\mathrm{S}}^{\left(1\right)}\left(t\right)\right\rangle \right|^{2}+\mathcal{O}\left(\left(\Omega t\right)^{4}\right).\label{eq:P_1_app_fin}
\end{align}
Here, we have exploited that $|\psi_{\mathrm{S}}^{\left(l\right)}\left(t\right)\rangle $
contains an even (odd) number of excitations for even (odd) order $l$. 
After insertion of Eqs.~(\ref{eq:psi_S_1})-(\ref{eq:psi_S_2}),
the expressions (\ref{eq:P_0_app_fin})-(\ref{eq:P_1_app_fin}) simplify
to the results (\ref{eq:pert_obs_P0})-(\ref{eq:pert_obs_P1}) shown in the main text. 
Next, we consider the site occupation operator $\hat{n}_{i}$ whose action on a state
$\left|0_{i}\right\rangle $ is $\hat{n}_{i}\left|0_{i}\right\rangle =0$.
Correspondingly, the order $l=0$ does not contribute to $n_{i}\left(t\right)$
and one finds 
\begin{align}
n_{i}\left(t\right) & = \sum_{l,l^{\prime}=1}^{\infty}\Big\langle \psi_{\mathrm{S}}^{(l^{\prime})}\left(t\right)\Big|\hat{n}_{i}\Big|\psi_{\mathrm{S}}^{\left(l\right)}\left(t\right)\Big\rangle
=\mathcal{O}\left(\left(\Omega t\right)^{2}\right)\nonumber \\
 & =\left\langle \psi_{\mathrm{S}}^{\left(1\right)}\left(t\right)\right|\hat{n}_{i}\left|\psi_{\mathrm{S}}^{\left(1\right)}\left(t\right)\right\rangle +\mathcal{O}\left(\left(\Omega t\right)^{4}\right),\label{eq:n_i_app_fin}
\end{align}
with the evaluated expression shown in Eq.~(\ref{eq:pert_obs_ni}). 
Lastly, we compute the correlation matrix $C_{i,j}\left(t\right)$. 
Surely, its diagonal elements $C_{i,i}\left(t\right)=n_{i}\left(t\right)$ are described by Eq.~(\ref{eq:n_i_app_fin}). 
Since $\hat{n}_{i}\hat{n}_{j}\,|\psi_{\mathrm{S}}^{\left(l=\left\{ 0,1\right\} \right)}\left(t\right)\rangle =0$,
for the off-diagonal elements one obtains
\begin{align}
C_{i,j\neq i}\left(t\right) & = \sum_{l,l^{\prime}=2}^{\infty}\Big\langle \psi_{\mathrm{S}}^{(l^{\prime})}\left(t\right)\Big|\hat{n}_{i}\hat{n}_{j}\Big|\psi_{\mathrm{S}}^{\left(l\right)}\left(t\right)\Big\rangle
=\mathcal{O}\left(\left(\Omega t\right)^{4}\right),\label{eq:C_ij_app_fin}
\end{align}
i.e. for the employed expansion order, the off-diagonal elements $C_{i,j}\left(t\right)$ can be neglected.

\section{Classical graph kernels\label{app:classical_graph_kernel}}

To deal with the classical kernel methods introduced in Subsec.~\ref{subsec:classical_kernels}, we employ external implementations for standard graph kernels, namely the Python library GraKeL. 
Here, we report the parameter and performance choices that are made to produce the benchmarks listed in table~\ref{tab:core_table}.

For all kernels, we choose to normalize the kernel matrix. This improves numerical stability and makes the subsequent optimization less sensitive to an arbitrary scaling of the kernel. In GraKeL, this corresponds to the \textit{normalize} option, which applies the following transformation:
\begin{align}
 K_{ij}^{\text{norm}} = K_{ij} \big/ \sqrt{K_{ii}K_{jj}}, 
\label{eq:normalization_classical_kernel}
\end{align}
where in principle, it should hold $K_{ii}>0$.
In practice, there can be special cases where $K_{ii}=0$ or $K_{ii}<0$ due to anomalies in the dataset or to numerical instabilities.
In such a case, we regularize the kernel matrix before normalizing via $K\rightarrow K+\epsilon \mathbbm{1}$ with $\epsilon \ll1$.
Such a case occurs for the \textit{Graphlet Subsampling} kernel on the PTC\_FM dataset where we choose $\epsilon = 10^{-3}$ for regularization.

For the Random Walk kernel, we use the \textit{exponential} variant rather than the \textit{geometric} one. 
The geometric formulation relies on the inversion of a matrix $(\mathbbm{1} - \lambda A_{\times})$, whose numerical stability is highly sensitive to the hyperparameter $\lambda$.
For this reason, we employ the exponential formulation which, although computationally more expensive (complexity $\mathcal{O}(N^6)$ versus $\mathcal{O}(N^3)$ for the geometric variant), is substantially more stable in practice.

Unless stated otherwise, all hyperparameters are left to their default values. 
We only report non-default settings below.
\begin{itemize}
  \item \textbf{Random Walk} (\texttt{kernel\_type=exponential}):
    \begin{itemize}
      \item MUTAG: lamda=5
      \item PTC\_FM$^{*}$: lamda=10.
    \end{itemize}

  \item \textbf{Graphlet Sampling} (\texttt{normalize=False}):
    \begin{itemize}
      \item MUTAG: $k=9$
      \item PTC\_FM$^{*}$: $k=7$.
    \end{itemize}
    \item \textbf{ShortestPath} (\texttt{algorithm\_type=floyd\_warshall}).
\end{itemize}

\section{Improved graph discrimination using local detuning\label{app:proof_prop1}}

In this Appendix we demonstrate the claim of proposition~\ref{prop:prop1} in the
main text: 
\textit{A quantum evolution kernel} $\kappa^{\mathrm{loc}}\left(\mathcal{G}_{\mu},\mathcal{G}_{\mu^{\prime}}\right)$\textit{
derived from a local-detuning Hamiltonian} $\hat{\mathcal{H}}_{\mathrm{loc}}^{\mathcal{G}}$ (\ref{eq:H_local}) \textit{is genuinely more expressive than its counterpart} $\kappa^{\mathrm{glob}}\left(\mathcal{G}_{\mu},\mathcal{G}_{\mu^{\prime}}\right)$\textit{
derived from the global-detuning Hamiltonian} $\hat{\mathcal{H}}_{\mathrm{glob}}^{\mathcal{G}}$ (\ref{eq:H_global})  \textit{for dinstinguishing attributed graphs.}

This proposition is demonstrated in two steps. (i) We show that there
exist non-isomorphic graphs which can be distinguished by the local
kernel, \textit{i.e.} $\kappa_{\mathrm{loc}}\left(\mathcal{G}_{\mu},\mathcal{G}_{\mu^{\prime}}\right)\neq0$,
but not by the global counterpart, \textit{i.e.} $\kappa_{\mathrm{glob}}\left(\mathcal{G}_{\mu},\mathcal{G}_{\mu^{\prime}}\right)=0$.
(ii) We show that the converse is always given in general, \textit{i.e.} if
the global kernel is able to distinguish two non-isomorphic graphs
{[}$\kappa_{\mathrm{glob}}\left(\mathcal{G}_{\mu},\mathcal{G}_{\mu^{\prime}}\right)\neq0${]},
then the local kernel is able to do the same {[}$\kappa_{\mathrm{loc}}\left(\mathcal{G}_{\mu},\mathcal{G}_{\mu^{\prime}}\right)\neq0${]}.
Here, the term \textit{in general} implies that accidental choices
of specific detuning values causing the local kernel to vanish can
never be ruled out. We argue however, that such accidental zeros form
a lower-dimensional manifold in parameter space and are thus non-generic.
In practical terms, it implies that a given (non-fine-tuned) choice
of local detuning values will almost surely avoid the accidental manifold.

(i) Here, it is sufficient to present one pair of non-isomorphic graphs
which satisfy the requirement, \textit{i.e.} the local kernel can distinguish
them but the global one can not. This situation arises for any two
graphs which have identical edges (and node positions) but different
node features. Since the local kernel encodes node features it would
be able to distinguish such graphs whereas for the global kernel the
two graphs are indistinguishable. In the context of the studied molecules
the node features that were associated with the local detuning field
were the atomic masses, \textit{cf.} Eq. (\ref{eq:local_detuning}). One example
could be two molecules with identical (or nearly identical) atomic
positions, but consisting of different atoms, such as nitrogen monoxide
$\mathrm{NO}$ {[}$a_{\mathrm{NO}}=115\,\mathrm{pm}${]} and carbon
monoxide $\mathrm{CO}$ {[}$a_{\mathrm{CO}}=113\,\mathrm{pm}${]}.
The bond distance is negligible but the atomic masses, $m_{\mathrm{N}}=14.01\,\mathrm{u}$
and $m_{\mathrm{C}}=12.01\,\mathrm{u}$, allow for a robust distinction.

(ii) In order to demonstrate the second point, we first recall some
notation. We start from the general quantum evolution kernel Eq. (\ref{eq:general_quantum_evolution_kernel}),
\begin{align}
K_{\mu\mu^{\prime}} & =F_{\kappa}\left[\boldsymbol{O}^{\mathcal{G}_{\mu}}\left(\mathbb{J}^{\mathcal{G}_{\mu}},\Omega,\boldsymbol{\delta}\right),\boldsymbol{O}^{\mathcal{G}_{\mu^{\prime}}}\left(\mathbb{J}^{\mathcal{G}_{\mu^{\prime}}},\Omega,\boldsymbol{\delta}^{\prime}\right)\right].\label{eq:generic_QK_appendix}
\end{align}
For convenience, we define the observable vectors $\boldsymbol{O}^{\mathcal{G}_{\mu}}\left(\mathbb{J}^{\mathcal{G}_{\mu}},\Omega,\boldsymbol{\delta}\right)$
as functions of Hamiltonian parameters (\ref{eq:Rydberg_Ham}) depending
on which Hamiltonian they were subjected to, \textit{cf.} Eqs. (\ref{eq:H_global}),
and (\ref{eq:H_local}). This allows us to define the function
\begin{align}
F\left(\boldsymbol{\delta}^{\mathrm{tot}}\right) & =F_{\kappa}\left[\boldsymbol{O}^{\mathcal{G}_{\mu}}\left(\mathbb{J}^{\mathcal{G}_{\mu}},\Omega,\boldsymbol{\delta}\right),\boldsymbol{O}^{\mathcal{G}_{\mu^{\prime}}}\left(\mathbb{J}^{\mathcal{G}_{\mu^{\prime}}},\Omega,\boldsymbol{\delta}^{\prime}\right)\right],\label{eq:F_func_delta_tot}
\end{align}
which exclusively depends on the local detuning with the $d_{\mathrm{max}}^{\mathrm{tot}}$-dimensional
vector $\boldsymbol{\delta}^{\mathrm{tot}}=\left(\boldsymbol{\delta},\boldsymbol{\delta}^{\prime}\right)$
and $d_{\mathrm{max}}^{\mathrm{tot}}=\left|\mathcal{V}\right|+\left|\mathcal{V}^{\prime}\right|$.
The function $F\left(\boldsymbol{\delta}^{\mathrm{tot}}\right)$ is
defined on the domain $\Gamma:=\left\{ \boldsymbol{\delta}^{\mathrm{tot}}\in\mathbb{R}^{d_{\mathrm{max}}^{\mathrm{tot}}}\right\} $.
Similarly as in the global-detuning Hamiltonian (\ref{eq:H_global}),
we define the total global detuning vector as $\boldsymbol{\delta}_{\mathrm{glob}}^{\mathrm{tot}}=\left(\delta_{0},\dots,\delta_{0}\right)$.
Now, we can recast the initial assumption wherein the global kernel
is able to distinguish the two non-isomorphic graphs as 
\begin{align}
F\left(\boldsymbol{\delta}_{\mathrm{glob}}^{\mathrm{tot}}\right) & \neq0.\label{eq:F_delta_glob_tot-1-1}
\end{align}

At this state we employ the identity theorem. Let us assume that $F_{K}$
is a real-analytic function and the same is the case for any observable
$\boldsymbol{O}^{\mathcal{G}_{\mu}}\left(\mathbb{J}^{\mathcal{G}_{\mu}},\Omega,\boldsymbol{\delta}\right)$;
the latter is discussed at the end of this Appendix. Given these assumptions,
it follows that the function $F\left(\boldsymbol{\delta}^{\mathrm{tot}}\right)$
{[}Eq. (\ref{eq:F_func_delta_tot}){]} is also real-analytic on the
entire domain $\Gamma$, and thus, the identity theorem is applicable.
We use a proof by contradiction. Let us assume that there exists an
open set on the domain $\Gamma$ (\textit{i.e.} a manifold with full dimension
$d_{\mathrm{max}}^{\mathrm{tot}}$) on which $F\left(\boldsymbol{\delta}^{\mathrm{tot}}\right)=0$.
The identity theorem would then require that $F\left(\boldsymbol{\delta}^{\mathrm{tot}}\right)=0$
on the entire domain $\Gamma$. This, however, contradicts Eq. (\ref{eq:F_delta_glob_tot-1-1})
since $\boldsymbol{\delta}_{\mathrm{glob}}^{\mathrm{tot}}$ lies within
the domain. As a result, no such open sets can exist on the domain
$\Gamma$. Manifolds on which the function $F\left(\boldsymbol{\delta}^{\mathrm{tot}}\right)$
vanishes are at most of dimension $d_{\mathrm{max}}^{\mathrm{tot}}-1$.
This circumstance makes it very unlikely that a generic choice of
detuning parameters fall onto such a manifold.

Lastly, we discuss the analyticity of observables $O\left(t,\boldsymbol{\lambda}\right)=\left\langle \psi\left(t,\boldsymbol{\lambda}\right)\right|\hat{O}\left|\psi\left(t,\boldsymbol{\lambda}\right)\right\rangle $
dependent on time $t$ and Hamiltonian parameters $\boldsymbol{\lambda}$,
in our case $\boldsymbol{\lambda}=\left(\mathbb{J},\Omega,\boldsymbol{\delta}\right)$.
Under fairly general assumptions, if the Hamiltonian $\hat{\mathcal{H}}\left(t,\boldsymbol{\lambda}\right)$
and the initial state $\left|\psi\left(0,\boldsymbol{\lambda}\right)\right\rangle $
depend analytically on parameters $\boldsymbol{\lambda}$, then the
time-evolved state $\left|\psi\left(t,\boldsymbol{\lambda}\right)\right\rangle =\hat{U}\left(t,\boldsymbol{\lambda}\right)\left|\psi\left(0,\boldsymbol{\lambda}\right)\right\rangle $
and all standard observable expectation values are also real-analytic
functions of $\boldsymbol{\lambda}$~\cite{Reed1978,Kato1995}.
Here, $\hat{U}\left(t,\boldsymbol{\lambda}\right)$ is the time-evolution
operator.

The question is really in which cases could the observable be non-analytic.
In the following we break down $5$ possibilities where it could occur,
and argue that none of these scenarios apply to our case. 

(1) Undoubtedly, if either the Hamiltonian $\hat{\mathcal{H}}\left(t,\boldsymbol{\lambda}\right)$
or the initial state $\left|\psi\left(0,\boldsymbol{\lambda}\right)\right\rangle $
depend non-analytically on the parameters, the non-analyticity would
be propagated to any observable. For example, an initial state could
be prepared as a ground state of $\hat{\mathcal{H}}\left(0,\boldsymbol{\lambda}\right)$,
and be non-analytically due to a level crossing. In our setup, neither
is the case. The Hamiltonian is a smooth function of its parameters,
and the initial state is quite featureless; it is the all-zero state
$\left|\psi\left(0,\boldsymbol{\lambda}\right)\right\rangle =\left|\boldsymbol{0}\right\rangle $.

(2) Non-analytic behavior is encountered in the context of phase transitions:
order-parameter kinks, and jumps, etc.. Such non-analyticities appear
across the whole phase transition spectrum: thermal transitions, quantum
phase transitions, dynamical phase transitions, topological phase
transitions, or even many-body localization (MBL) transitions, regardless
whether they are classified as first-order or second-order transitions.
Nonetheless, true non-analyticities only appear in the thermodynamic
limit (for infinitely large system sizes). For finite-size systems,
such as the graphs which we are interested in, no non-analytic behavior
is caused by phase transitions. All quantities are analytic and evolve
smoothly across potential transition regions.

(3) Non-analytic behavior can also appear inside a MBL phase (additionally
to the transition boundary non-analyticities) or within certain glassy
systems. Again, these features are only genuinely non-analytic in
the thermodynamic limit.

(4) Non-standard observables such as post-selected observables could
also cause non-analytic behavior. In the context of post-selected observables,
the time-evolved state is first projected onto a specific subspace
of the Hilbert space before an observable operator is applied. This
can trigger a non-analytic behavior, for example if the state has
zero weight in the respective subspace, leading to a divergence due
to the normalization. In the presented protocol, only standard observables
are implemented, see Eq. (\ref{eq:observables}).

(5) Non-analyticities can also appear in the context of level crossings.
Here, however, they only show for specific observables. For example,
if one tracks an observable for a specific eigenstate (such as the
ground state). For our implemented observables, level crossings do
not cause non-analytic behavior.

\section{Properties of $\mathrm{GDQC}$\label{app:GDQC}}

In this appendix we prove the proposition~\ref{prop:prop2}: \textit{
The function $\kappa_{\mathrm{GDQC}}$ defined in Eq.~(\ref{eq:GDQC_t}) is a valid graph kernel, i.e. it is symmetric positive semi-definite on the space of graphs equipped with any graph distance $d_{\mathcal{G}}$, and is isomorphism-invariant such that $\kappa_{\mathrm{GDQC}}(\mathcal{G}_{\mu}, \mathcal{G}_{\mu^\prime}) = 1$ whenever $\mathcal{G}_\mu$ and $\mathcal{G}_{\mu^\prime}$ are isomorphic.
}
Let us recall that the generalized-distance quantum-correlation kernel 
$\kappa_{\mathrm{GDQC}}$ (\ref{eq:GDQC_t}) is defined via a scalar product 
\begin{align}
K_{\mu\mu^{\prime}}^{\mathrm{GDQC}} & =\kappa_{\mathrm{GDQC}}\left(\mathcal{G}_{\mu},\mathcal{G}_{\mu^{\prime}}\right)\label{eq:K_GDQC_def_App}\\
 & =\boldsymbol{\chi}^{\intercal}\left[C^{\mathcal{G}_{\mu}}\left(t\right),D^{\mathcal{G}_{\mu}}\right]\,\boldsymbol{\chi}\left[C^{\mathcal{G}_{\mu^{\prime}}}\left(t\right),D^{\mathcal{G}_{\mu^{\prime}}}\right]\nonumber \\
 & =\sum_{\nu = 1}^{N_\chi}\chi_{\nu}\left[C^{\mathcal{G}_{\mu}}\left(t\right),D^{\mathcal{G}_{\mu}}\right]\,\chi_{\nu}\left[C^{\mathcal{G}_{\mu^{\prime}}}\left(t\right),D^{\mathcal{G}_{\mu^{\prime}}}\right],\nonumber 
\end{align}
where $C^{\mathcal{G}}\left(t\right)$ denotes the correlation
matrix (\ref{eq:observables}) associated with the graph $\mathcal{G}$,
\textit{cf.} Eqs. (\ref{eq:general_quantum_evolution_kernel}), (\ref{eq:H_global}),
(\ref{eq:H_local}), and $D^{\mathcal{G}}$ is the generalized-distance
matrix. The index iterates over the whole dataset of $N_{\mathrm{data}}$
graphs, \textit{i.e.} $\mu,\mu^{\prime}\in\left\{ 1,2,\dots,N_{\mathrm{data}}\right\} $.
Naturally, the correlation matrix $C^{\mathcal{G}}\left(t\right)$
is a real-valued observable since it is derived from a Hermitian operator.
Consequently, each element of the $N_\chi$-dimensional $\mathrm{GDQC}$ feature vector $\boldsymbol{\chi}\left[C^{\mathcal{G}}\left(t\right),D^{\mathcal{G}}\right]$ is real-valued with $N_\chi = N_{\mathrm{bins}}^CN_{\mathrm{bins}}^D$. 
By continuation, the $\mathrm{GDQC}$ kernel $\kappa_{\mathrm{GDQC}}\left(\mathcal{G}_{\mu},\mathcal{G}_{\mu^{\prime}}\right)$
is also real-valued. 
Next, we prove the 3 statements of proposition~\ref{prop:prop2} point by point.
\begin{proof}
(1) A kernel defined via a scalar product like $\kappa_{\mathrm{GDQC}}$ (\ref{eq:K_GDQC_def_App}) is symmetric, \textit{i.e.} $\kappa_{\mathrm{GDQC}}\left(\mathcal{G}_{\mu},\mathcal{G}_{\mu^{\prime}}\right)=\kappa_{\mathrm{GDQC}}\left(\mathcal{G}_{\mu^{\prime}},\mathcal{G}_{\mu}\right)$.
Moreover, a kernel is positive semi-definite if for any real
vector $\boldsymbol{c}\in\mathbb{R}^{N_{\mathrm{data}}}$ it holds
$\boldsymbol{c}^{\intercal}K^{\mathrm{GDQC}}\boldsymbol{c}\geq0$.
An explicit computation directly leads to 
\begin{align}
\boldsymbol{c}^{\intercal}K^{\mathrm{GDQC}}\boldsymbol{c} & =\sum_{\nu=1}^{N_\chi}\left(\sum_{\mu}c_{\mu}\chi_{\nu}\left[C^{\mathcal{G}_{\mu}}\left(t\right),D^{\mathcal{G}_{\mu}}\right]\right)^{2}\geq0,\label{eq:K_GDQC_PSD}
\end{align}
confirming that the $\mathrm{GDQC}$ kernel is indeed positive semi-definite.\\
(2). We demonstrate that the kernel is invariant under node
label permutations. Hereby, it suffices to show that the $\mathrm{GDQC}$ feature vector satisfies 
\begin{align}
\boldsymbol{\chi}\left[C^{\mathcal{G}^{\Pi}}\left(t\right),D^{\mathcal{G}^{\Pi}}\right] & =\boldsymbol{\chi}\left[C^{\mathcal{G}}\left(t\right),D^{\mathcal{G}}\right],\label{eq:chi_U_Pi_transform}
\end{align}
with $\mathcal{G}^{\Pi}$ being a graph which has its nodes permuted by a given permutation $\Pi$. 
Specifically, we demonstrate that the correlation and distance matrices of a permuted graph are related to their non-permuted counterparts via the relations
\begin{align}
C_{i,j}^{\mathcal{G}^{\Pi}} & = C_{\Pi^{-1}\left(i\right),\Pi^{-1}\left(j\right)}^{\mathcal{G}},
& D_{i,j}^{\mathcal{G}^{\Pi}} & = D_{\Pi^{-1}\left(i\right),\Pi^{-1}\left(j\right)}^{\mathcal{G}}.\label{eq:CD_permutation}
\end{align}
Consequently, the permutation of a graph is identical to the inverse permutation of the matrix indices associated with the non-permuted graph.
Such a transformation, however, has no effect on the $\mathrm{GDQC}$ feature vector $\boldsymbol{\chi}\left[C^{\mathcal{G}}\left(t\right),D^{\mathcal{G}}\right]$,
as the feature-vector elements are simply counts of the number of
values falling into a predefined interval, \textit{cf.} algorithm~\ref{algo:GDQC} and Fig.~\ref{fig:GDQC_illustration}. 
Thus, the $\mathrm{GDQC}$ kernel $\kappa_{\mathrm{GDQC}}$ (\ref{eq:K_GDQC_def_App}) is invariant under node label permutations. 
It remains to prove that the transformation relations (\ref{eq:CD_permutation}) are actually true.

Concerning the distances it is clear that the distance matrices associated with $\mathcal{G}$ and $\mathcal{G}^{\Pi}$ are related via $D_{i,j}^{\mathcal{G}}=D_{\Pi\left(i\right),\Pi\left(j\right)}^{\mathcal{G}^{\Pi}}$ leading to the relation in Eq.~(\ref{eq:CD_permutation}). 
Concerning the correlation matrix relation, a few mathematical steps are in order.
We begin by considering the graph-embedded coupling matrix $\mathbb{J}^{\mathcal{G}_{\mu}}$,
Eq. (\ref{eq:graph_dependent_J}), and the local detuning field $\boldsymbol{\delta}^{\mathcal{G}}$,
Eq. (\ref{eq:local_detuning}). One directly obtains the relations
\begin{align}
\mathbb{J}_{i,j}^{\mathcal{G}^{\Pi}} & =\mathbb{J}_{\Pi^{-1}\left(i\right),\Pi^{-1}\left(j\right)}^{\mathcal{G}}, & \delta_{i}^{\mathcal{G}^{\Pi}} & =\delta_{\Pi^{-1}\left(i\right)}^{\mathcal{G}},\label{eq:graph_dep_coupl_matrix_permuted}
\end{align}
connecting the quantities between the permuted and the non-permuted graphs. 
Note that the relation (\ref{eq:graph_dep_coupl_matrix_permuted}) for the detuning is also trivially satisfied for the global detuning case.
Next, we define the unitary operator $\hat{U}_{\Pi}$ which permutes the atoms according to $\Pi$. 
Its action on a basis state $\left|\boldsymbol{b}\right\rangle$ is given by 
\begin{align}
\hat{U}_{\Pi}\left|\boldsymbol{b}\right\rangle  & =\hat{U}_{\Pi}\left|b_{1}b_{2}\cdots b_{N}\right\rangle =\left|b_{\Pi\left(1\right)}b_{\Pi\left(2\right)}\cdots b_{\Pi\left(N\right)}\right\rangle ,\label{eq:U_Pi_action}
\end{align}
and hence, it transforms the local operators in the Hamiltonian [Eq.~(\ref{eq:Rydberg_Ham})] according to
\begin{align}
\hat{U}_{\Pi}\hat{n}_{i}\hat{U}_{\Pi}^{\dagger} & =\hat{n}_{\Pi\left(i\right)}, & \hat{U}_{\Pi}\hat{\sigma}_{i}^{x}\hat{U}_{\Pi}^{\dagger} & =\hat{\sigma}_{\Pi\left(i\right)}^{x}.\label{eq:n_sx_U_Pi_transform}
\end{align}
Combining the properties (\ref{eq:graph_dep_coupl_matrix_permuted})
and (\ref{eq:n_sx_U_Pi_transform}), we obtain the Hamiltonian [Eq.~(\ref{eq:Rydberg_Ham})] relation
\begin{align}
\hat{\mathcal{H}}\left(\mathbb{J}^{\mathcal{G}^{\Pi}},\Omega,\delta^{\mathcal{G}^{\Pi}}\right) & =\hat{U}_{\Pi}\,\hat{\mathcal{H}}\left(\mathbb{J}^{\mathcal{G}},\Omega,\delta^{\mathcal{G}}\right)\,\hat{U}_{\Pi}^{\dagger},\label{eq:Ham_U_Pi_transform}
\end{align}
where in between we switched the summation indices $i\rightarrow i^{\prime}=\Pi^{-1}\left(i\right)$ before applying Eq.~(\ref{eq:n_sx_U_Pi_transform}). 
Correspondingly, the time-evolved state (\ref{eq:psi_t}),
\begin{align}
\left|\psi_{\mathcal{G}}\left(t\right)\right\rangle  & =\exp\left[-\mathsf{i}(t/\hbar)\,\hat{\mathcal{H}}\left(\mathbb{J}^{\mathcal{G}},\Omega,\boldsymbol{\delta}^{\mathcal{G}}\right)\right]\,\left|\psi\left(0\right)\right\rangle ,\label{eq:psi_t_App}
\end{align}
satisfies the transformation relation
\begin{align}
\left|\psi_{\mathcal{G}^{\Pi}}\left(t\right)\right\rangle  & =\hat{U}_{\Pi}\,\left|\psi_{\mathcal{G}}\left(t\right)\right\rangle ,\label{eq:psi_t_U_Pi_transform}
\end{align}
where we exploited the permutation invariance of the initial state
$\left|\psi\left(0\right)\right\rangle =\left|\boldsymbol{0}\right\rangle $, \textit{i.e.} $\hat{U}_{\Pi}^{\dagger}\left|\boldsymbol{0}\right\rangle =\left|\boldsymbol{0}\right\rangle $.
Lastly, we insert the relation (\ref{eq:psi_t_U_Pi_transform}) into the correlation matrix [Eq.~(\ref{eq:observables})], 
\begin{align}
C_{i,j}^{\mathcal{G}}\left(t\right) & =\left\langle \psi_{\mathcal{G}}\left(t\right)\big|\,\hat{n}_{i}\,\hat{n}_{j}\,\big|\psi_{\mathcal{G}}\left(t\right)\right\rangle ,\label{eq:correlation_mat_App}
\end{align}
leading to the final transformation relation (\ref{eq:CD_permutation}),
\begin{align}
C_{i,j}^{\mathcal{G}^{\Pi}}\left(t\right) & =\left\langle \psi_{\mathcal{G}}\left(t\right)\big|\hat{n}_{\Pi^{-1}\left(i\right)}\hat{n}_{\Pi^{-1}\left(j\right)}\big|\psi_{\mathcal{G}}\left(t\right)\right\rangle \nonumber \\
 & =C_{\Pi^{-1}\left(i\right),\Pi^{-1}\left(j\right)}^{\mathcal{G}}\left(t\right).\label{eq:corr_mat_U_Pi_transform}
\end{align}
(3) To prove the last property we define a pair of isomorphic graphs
$\mathcal{G}_{1}$ and $\mathcal{G}_{2}=\mathcal{G}_{1}^{\Pi}$. It
remains to show the property 
\begin{align}
\kappa_{\mathrm{GDQC}}\left(\mathcal{G}_{1},\mathcal{G}_{1}^{\Pi}\right) & =1,\label{eq:G1_G1Pi_kernel_property}
\end{align}
for the kernel function (\ref{eq:K_GDQC_def_App}). The property (\ref{eq:G1_G1Pi_kernel_property})
follows trivially from the $\mathrm{GDQC}$ feature vector invariance (\ref{eq:chi_U_Pi_transform}),
and the normalization of the feature vector $\left\Vert \boldsymbol{\chi}\left[C^{\mathcal{G}}\left(t\right),D^{\mathcal{G}}\right]\right\Vert =1$.
\end{proof}

\section{Expressiveness of $\mathrm{GDQC}$\label{app:expressivity_gdqc}}

In this Appendix we prove the proposition~\ref{prop:prop3}: \textit{
The $\mathrm{GDQC}$ kernel $\kappa_{\mathrm{GDQC}}$ matches the expressiveness of the GD-WL refinement~\cite{zhang2023rethinking} when GD-WL is initialized with node colors constructed from the multisets of (binned) correlation values in the rows/columns of $C^{\mathcal{G}}(t)$.
}
\begin{proof}
Consider the correlation matrix $C_{i,j}^{\mathcal{G}}\left(t\right)$,
Eq. (\ref{eq:correlation_mat_App}). First, let's start by giving
an example of how one can use the correlation matrix (at a fixed time
step $t$) as a color refinement algorithm. 
Let 
\begin{align}
q_{t}\left(i\right) & =\mathrm{hash}\left(\left\{ \!\!\left\{ C_{i,j}^{\mathcal{G}}\left(t\right)\;\;;\;\;j\in\mathcal{V}\right\} \!\!\right\} \right),\label{eq:q_t_i}
\end{align}
be the color attributed to node $i$ at time step $t$ of the dynamics,
where hash represents any injective function on its domain. Recall
that GD-WL test can be formalized as : 
\begin{align}
\chi_{\mathcal{G}}^{(l)}\left(j\right) & =\mathrm{hash}\left(\left\{ \!\!\left\{ \left(D_{i,j}^{\mathcal{G}},\chi_{\mathcal{G}}^{(l-1)}\left(i\right)\right)\;\;;\;\;i\in\mathcal{V}\right\} \!\!\right\} \right).\label{eq:hash}
\end{align}
In general, one starts with a uniform $\chi_{\mathcal{G}}^{(0)}\left(j\right)$,
but in our case we put : 
\begin{align}
\chi_{\mathcal{G}}^{(0)}\left(j\right) &=q_{t}\left(j\right).\label{eq:chi_G_0}
\end{align}
Now, consider the following tensor 
\begin{align}
M_{i,j}^{\mathcal{G}} & =\left(D_{i,j}^{\mathcal{G}},\chi_{\mathcal{G}}^{(0)}\left(j\right)\right).\label{eq:M_G_ij}
\end{align}
In order to prove the statement of proposition~\ref{prop:prop3}, we start by showing
the following lemma : 
\begin{lemma} \label{lemma : permutation}
Two graphs $\mathcal{G}_{1}$ and $\mathcal{G}_{2}$ are undistinguishable
by the GD-WL test if and only if there is a permutation $\Pi$ such that for every pair (i,j) we have:
\begin{align}
M^{\mathcal{G}_{1}}_{i,j} &= M^{\mathcal{G}_{2}}_{\Pi(i),\Pi(j)}
.\label{eq:MG1_MG2_relation}
\end{align}
\end{lemma} 
\begin{proof}
First we prove the \textit{if} part. Suppose such a permutation exists,
and call 
\begin{align}
R_{\mathcal{G}}^{(l)}\left(i\right) & =\left\{ \!\!\left\{ \left(D_{i,j}^{\mathcal{G}},\chi_{\mathcal{G}}^{(l)}\left(j\right)\right)\;\;:\;\;j\in\mathcal{V}\right\} \!\!\right\} ,\label{eq:RG_l_i}
\end{align}
the multi-set for node $i$ corresponding to the data hashed
in round $l$. Because we assume that this permutation exists,
one must have : 
\begin{align}
R_{\mathcal{G}_{1}}^{(0)}\left(i\right) & =\left\{ \!\!\left\{ M_{i,j}^{\mathcal{G}_{1}}\;\;:\;\;j\in\mathcal{V}_{1}\right\} \!\!\right\} 
  =\left\{ \!\!\left\{ M_{\Pi\left(i\right),w}^{\mathcal{G}_{2}}\;\;:\;\;w\in\mathcal{V}_{2}\right\} \!\!\right\} \nonumber \\
 & =R_{\mathcal{G}_{2}}^{(0)}\left(\Pi\left(i\right)\right),\label{eq:R_G1_R_G2_relation}
\end{align}
with $w=\Pi\left(j\right)$. This ensures that at the first round,
vertex $i$ in $\mathcal{G}_{1}$ and vertex $\Pi\left(i\right)$
in $\mathcal{G}_{2}$ feed identical multi-sets into the hash function
{[}\textit{cf.} Eq.~(\ref{eq:hash}){]}, and therefore receive the same color.
Since by hypothesis we have $D_{i,j}^{\mathcal{G}_{1}}=D_{\Pi\left(i\right),\Pi\left(j\right)}^{\mathcal{G}_{2}}$
for all $(i,j)$, we can also see that the multi-set hashed in the
next round: 
\begin{align}
R_{\mathcal{G}_{1}}^{(1)}\left(i\right) & =\left\{ \!\!\left\{ \left(D_{i,j}^{\mathcal{G}_{1}},\chi_{\mathcal{G}_{1}}^{(1)}\left(j\right)\right)\;\;:\;\;j\in\mathcal{V}_{1}\right\} \!\!\right\} \nonumber \\
 & =\left\{\!\!\left\{ \left(D_{\Pi\left(i\right),\Pi\left(j\right)}^{\mathcal{G}_{2}},\chi_{\mathcal{G}_{2}}^{(1)}\left(\Pi\left(j\right)\right)\right)\;\;:\;\;j\in\mathcal{V}_{1}\right\}\!\!\right\} \nonumber \\
 & =R_{\mathcal{G}_{2}}^{(1)}\left(\Pi\left(i\right)\right),\label{eq:R1_G1_R1_G2_relation}
\end{align}
leading to $\chi_{\mathcal{G}_{1}}^{(2)}\left(i\right)=\chi_{\mathcal{G}_{2}}^{(2)}\left(\Pi\left(i\right)\right)$.
The process repeats indefinitely itself until the number of iterations
reaches its maximal value $l_{\mathrm{max}}$, and we thus have identical
color classes on $\mathcal{G}_{1}$ and $\mathcal{G}_{2}$, which
means undistinguishability. Now we prove the \textit{only if} part,
\textit{i.e} when the GD-WL process never separates the graphs, then
such a permutation $\Pi$ must exist. We know that after at most $n$
iterations ($n$ being the number of nodes in $\mathcal{G}_{1}$ and
$\mathcal{G}_{2}$) the final colors, noted $\chi^{*}$ appears the same number of times in each graph as they are assumed
to be undistinguishable. If for each color class we choose a one-to-one
correspondence between the vertices of $\mathcal{G}_{1}$ and those
of $\mathcal{G}_{2}$ that share that color, we obtain a bijection
\begin{align}
\Pi & :\mathcal{V}_{1}\mapsto\mathcal{V}_{2}, & \mathrm{with} & : & \chi_{\mathcal{G}_{1}}^{*}\left(i\right) & =\chi_{\mathcal{G}_{2}}^{*}\left(\Pi\left(i\right)\right) & \forall & i.\label{eq:bijection}
\end{align}
We now show by (downward) induction that 
\begin{align}
\chi_{\mathcal{G}_{1}}^{(l)}\left(i\right) & =\chi_{\mathcal{G}_{2}}^{(l)}\left(\Pi\left(i\right)\right), & D_{i,j}^{\mathcal{G}_{1}} & =D_{\Pi\left(i\right),\Pi\left(j\right)}^{\mathcal{G}_{2}}, & \forall & \left(i,j\right),\label{eq:induction}
\end{align}
where $l\leq l_{\mathrm{max}}$. The base case is $l=l_{\mathrm{max}}$
which is exactly our~(\ref{eq:bijection}). Now assume that~(\ref{eq:induction})
holds, when colors at round $l$ are computed, vertex $i$ in $\mathcal{G}_{1}$
hashes $R_{\mathcal{G}_{1}}^{(l)}\left(i\right)$. By the induction
hypothesis we have $R_{\mathcal{G}_{1}}^{(l-1)}\left(i\right)=R_{\mathcal{G}_{2}}^{(l-1)}\left(\Pi\left(i\right)\right)$
which means that each ordered pair in $R_{\mathcal{G}_{1}}^{(l-1)}\left(i\right)$
occurs with the same frequency in $R_{\mathcal{G}_{2}}^{(l-1)}\left(\Pi\left(i\right)\right)$.
This allows us to align each tuple with distance $D_{i,j}^{\mathcal{G}_{1}}$
with the corresponding tuple $D_{\Pi\left(i\right),\Pi\left(j\right)}^{\mathcal{G}_{2}}$,
as the distances don't change from one round to the other, ensuring
that $\chi_{\mathcal{G}_{1}}^{(l-1)}\left(j\right)=\chi_{\mathcal{G}_{2}}^{(l-1)}\left(\Pi\left(j\right)\right)$
valid for $\forall j\in\mathcal{V}_{1}$. \\
One can then complete the downward induction to $l=0$ which shows
\begin{align}
M_{i,j}^{\mathcal{G}_{1}} & =\left(D_{i,j}^{\mathcal{G}_{1}},\chi_{\mathcal{G}_{1}}^{(0)}\left(j\right)\right)=\left(D_{\Pi\left(i\right),\Pi\left(j\right)}^{\mathcal{G}_{2}},\chi_{\mathcal{G}_{2}}^{(0)}\left(\Pi\left(j\right)\right)\right)\nonumber \\
 & =M_{\Pi\left(i\right),\Pi\left(j\right)}^{\mathcal{G}_{2}}.\label{eq:MG1_MG2_relation-1}
\end{align}
Rewriting this for every $(i,j)$ gives Eq.~(\ref{eq:MG1_MG2_relation}). 
\end{proof}

From here, the proof of proposition~\ref{prop:prop3} is straightforward. If two graphs
are undistinguishable by the GD-WL test with $q_{t}\left(i\right)$
as initial coloring, then even when the number of bins is large enough (with maximal expressiveness attained when $N_{\mathrm{bins}}^D>  \left[\min\left(D^\mathcal{G}_{i,j}-D^\mathcal{G}_{i^\prime,j^\prime}\right) \right]^{-1}$
and $N_{\mathrm{bins}}^C> \left[\min\left(C_{i,j}-C_{i^\prime,j^\prime}\right) \right]^{-1}$), according to lemma~\ref{lemma : permutation} their corresponding
tensors $M^{\mathcal{G}_{1}}$ and $M^{\mathcal{G}_{2}}$ are one
permutation away one from the other. Considering that hash is an injective function, this implies that, at a fixed graph distance value, the multi-sets obtained from the corresponding correlation matrix rows
are also the same, and thus the number of times a given value for
the correlation appears (while its corresponding nodes are at the
same fixed distance) is also the same. Given the definition of the $\mathrm{GDQC}$ kernel, the corresponding feature vectors for such a pair of graphs are identical, and their kernel value is equal to $1$, which completes the proof. 
\end{proof}

\section{Kernel pooling statement\label{app:kernel_pooling}}
In this appendix, we demonstrate the claim of the proposition~\ref{prop:prop4}: \textit{
Let $S=\{\kappa_1,\dots,\kappa_N\}$ be a set of graph kernels, and let $\kappa_A$ and $\kappa_B$ be kernels obtained by pooling two subsets $S_A$ and $S_B$ of $S$ using (\ref{eq:lin_comb}) or (\ref{eq:prod}) as pooling rules. If $S_A \subset S_B$, then $\kappa_B$ is at least as expressive as $\kappa_A$.
}

We restrict this proof to the case of the two pooling operations described in this paper, the sum pooling (\ref{eq:lin_comb}) and the element-wise product pooling (\ref{eq:prod}). 
It is enough to show that the expressiveness does not get lowered
by adjoining new kernels to the set of pooled ones.

\begin{proof}
A (graph) kernel $\kappa_{B}$ is more expressive than another (graph)
kernel $\kappa_{A}$ if there exists a given pair of non-isomorphic
graphs $\mathcal{G}_{1}$ \& $\mathcal{G}_{2}$ that $\kappa_{B}$
can distinguish but not $\kappa_{A}$. Furthermore, $\kappa_{B}$
must also be able to distinguish any such pair that $\kappa_{A}$
distinguishes. For normalized kernels, this translates to $\kappa_{A}(\mathcal{G}_{1},\mathcal{G}_{2})=1$
and $\kappa_{B}(\mathcal{G}_{1},\mathcal{G}_{2})<1$. 
Hereby, we adopt the convention that isomorphic graphs should have a maximal kernel value, which is 1 when the kernel is normalized. ove

First we need to fix the weight updating rule that ensures comparability
between the kernels $\kappa_{A}$ and $\kappa_{B}$ in the case of
sum pooling,
\begin{align}
\kappa_{A,+}\left(\mathcal{G}_{1},\mathcal{G}_{2}\right) & =\sum_{i=1}^{n_{A}}\alpha_{i}\,\kappa_{i}\left(\mathcal{G}_{1},\mathcal{G}_{2}\right),\label{eq:kappa_A+}
\end{align}
where $S_{A}=\{\kappa_{1},...,\kappa_{n_{A}}\}$ and $\sum_{i=1}^{n_{A}}\alpha_{i}=1$
which then becomes (after enriching the pooling set) : 
\begin{align}
\kappa_{B,+}\left(\mathcal{G}_{1},\mathcal{G}_{2}\right) & =\sum_{i=1}^{n_{A}}\alpha_{i}^{\prime}\,\kappa_{i}\left(\mathcal{G}_{1},\mathcal{G}_{2}\right)+\sum_{i=n_{A}+1}^{n_{B}}\alpha_{i}^{\prime}\,\kappa_{i}\left(\mathcal{G}_{1},\mathcal{G}_{2}\right),\label{eq:kappa_B+}
\end{align}
where $S_{B}=\{S_{A},\kappa_{n_{A}+1},...,\kappa_{n_{B}}\}$ and $\sum_{i=1}^{n_{B}}\alpha_{i}^{\prime}=1$.
To ensure the evaluation of the set adjoining operation, we have to
consider the case where the update of the weights only introduces
a factor that is proportional to the relative set sizes \textit{i.e.}
:
\begin{align}
\Big(\sum_{i=1}^{n_{A}}\alpha_{i}^{\prime}\Big)\big/\Big(\sum_{i=1}^{n_{B}}\alpha_{i}^{\prime}\Big) & =\left|S_{A}\right|\big/\left|S_{B}\right|,\label{eq:alpha_P_over_alpha}
\end{align}
which is obtained when $\alpha_{i}^{\prime}=\alpha_{i}\,n_{A}/n_{B}$.
Next we can list the two following cases : 
\begin{align}
\kappa_{+,B}\left(\mathcal{G}_{1},\mathcal{G}_{2}\right) & =\frac{n_{A}}{n_{B}}\kappa_{+,A}\left(\mathcal{G}_{1},\mathcal{G}_{2}\right)+\sum_{i=n_{A}+1}^{n_{B}}\alpha_{i}^{\prime}\kappa_{i}\left(\mathcal{G}_{1},\mathcal{G}_{2}\right),\label{eq:kappa_+B}
\end{align}
and 
\begin{align}
\kappa_{\times,B}\left(\mathcal{G}_{1},\mathcal{G}_{2}\right) & =\kappa_{\times,A}\left(\mathcal{G}_{1},\mathcal{G}_{2}\right)\prod_{i=n_{A}+1}^{n_{B}}\kappa_{i}\left(\mathcal{G}_{1},\mathcal{G}_{2}\right).\label{eq:kappa_xB}
\end{align}
So if $\kappa_{\times,A}=1$ then $\kappa_{\times,B}\leq1$ provided
that $\kappa_{i}\left(\mathcal{G}_{1},\mathcal{G}_{2}\right)$. 

Now if there is at least one kernel $\kappa_{i}\in\{k_{n_{A}+1},...,\kappa_{n_{B}}\}$
such that $\kappa_{i}\left(\mathcal{G}_{1},\mathcal{G}_{2}\right)<1$
then it directly follows from the two previous equations that $\kappa_{+,B}<\kappa_{+,A}$
and $\kappa_{\times,B}<\kappa_{\times,A}$. This assumption is not
necessarily true, and in the case where $\kappa_{i}\left(\mathcal{G}_{1},\mathcal{G}_{2}\right)=1\;\;\forall\;\;\kappa_{i}\in\{k_{n_{A}+1},...,\kappa_{n_{B}}\}$
then we have $\kappa_{+,B}=\kappa_{+,A}$ and $\kappa_{\times,B}=\kappa_{\times,A}$.
This allows us to state that without any specific knowledge of $\{k_{n_{A}+1},...,\kappa_{n_{B}}\}$
that $\kappa_{+,B}\leq\kappa_{+,A}$ and $\kappa_{\times,B}\leq\kappa_{\times,A}$.
\end{proof}

\bibliography{apssamp}
\end{document}